\documentclass[orivec]{llncs}
\usepackage{stmaryrd}
\usepackage[sectionbib,numbers,sort&compress]{natbib}

\usepackage{boxedminipage}
\usepackage{uc/uc}
\usepackage{cryptocode}
\usepackage{url}
\usepackage{color,soul}

\newcommand{\samuni}{\stackrel{\$}\gets}

\newcommand{\plpl}{\!+\hspace{-3pt}+}

\newcommand{\enscheme}{\mathrm{EN}}

\newcommand{\phyidealfunc}{\mathbb{R}}

\newcommand{\phyrecfields}{\mathit{A}}

\newcommand{\phyerrfunc}{\mathit{e}}
\newcommand{\phyusermeas}{\mathit{u}}
\newcommand{\phymeasret}{\phyusermeas^{\phyerrfunc}_{\phyrecfields}}
\newcommand{\phyrecinf}{\mathsf{INFECTED}}

\newcommand{\timeidealfunc}{\mathbb{T}}

\newcommand{\enidealfunc}{\funct{\enscheme}}
\newcommand{\enriskfunc}{\rho}
\newcommand{\enerrorfuncset}{\mathit{E}}
\newcommand{\enfakingfuncset}{\Phi}
\newcommand{\enleakfunc}{\mathcal{L}}

\newcommand{\ennoisyrec}{\widetilde{\mathrm{R}}_\mathit{\epsilon}}
\newcommand{\ennewnoisyrec}{\widetilde{\mathrm{R}}^*}
\newcommand{\ensharexprec}{\mathbf{SE}}

\newcommand{\enuser}{\mathit{U}}
\newcommand{\enusers}{\overline{\mathbf{U}}}
\newcommand{\enuserscorr}{\widetilde{\mathbf{U}}}
\newcommand{\enerrmsg}{\mathsf{error}}

\newcommand{\tbbidealfunc}{\mathcal{F}_{\mathsf{TBB}}}
\newcommand{\tbbrec}{\mathcal{C}}

\newcommand{\fesr}{\mathrm{FESR}}

\newcommand{\incubt}{\mathcal{T}}

\newcommand{\sigscheme}{\Sigma}
\newcommand{\sigg}{\mathsf{Gen}}
\newcommand{\sigs}{\mathsf{Sign}}
\newcommand{\sigv}{\mathsf{Vrfy}}
\newcommand{\sigpk}{\mathsf{vk}_\sigscheme}
\newcommand{\sigsk}{\mathsf{sk}_\sigscheme}
\newcommand{\irsig}[1]{\sigma_{\mathsf{#1}}}

\newcommand{\pke}{\mathsf{PKE}}
\newcommand{\pkeg}{\mathsf{PGen}}
\newcommand{\pkee}{\mathsf{Enc}}
\newcommand{\pked}{\mathsf{Dec}}
\newcommand{\pkepk}{\mathsf{pk}}
\newcommand{\pkesk}{\mathsf{sk}}
\newcommand{\pkemsg}{\mathrm{m}}
\newcommand{\pkerc}{\mathrm{r}}
\newcommand{\pkect}{\ct}

\newcommand{\nizkscheme}{\mathsf{N}}

\newcommand{\Prv}{\mathcal{P}}
\newcommand{\Ver}{\mathcal{V}}
\newcommand{\ncrs}{\mathsf{crs}}

\newcommand{\Extr}{\mathcal{E}}
\newcommand{\trap}{\tau}

\newcommand{\SC}{\mathcal{SC}}
\newcommand{\CRS}{\mathcal{CRS}}
\newcommand{\REP}{\mathcal{REP}}

\newcommand{\Gatt}{G_{\mathsf{att}}}

\newcommand{\gattmpk}{\mathsf{vk}}

\newcommand{\eid}[1]{\mathsf{eid}_{\mathsf{#1}}}
\newcommand{\prog}[1]{\mathsf{prog}_{\mathsf{#1}}}
\newcommand{\KME}{\mathsf{KME}}
\newcommand{\DE}{\mathsf{DE}}
\newcommand{\FE}{\mathsf{FE}}
\newcommand{\resume}{\mathsf{resume}}
\newcommand{\install}{\mathsf{install}}
\newcommand{\inp}{\mathsf{input}}
\newcommand{\outp}{\mathsf{output}}
\newcommand{\getpk}{\mathsf{getPK()}}
\newcommand{\initkw}{\mathrm{init}}
\newcommand{\insetupkw}{\mathrm{init}\textrm{-}\mathrm{setup}}
\newcommand{\cosetupkw}{\mathrm{complete}\textrm{-}\mathrm{setup}}
\newcommand{\compkw}{\mathrm{computed}}

\newcommand{\provisionkw}{\mathrm{provision}}
\newcommand{\runkw}{\mathrm{run}}
\newcommand{\funcKeyF}{\mathsf{sk}_{\funcF}}

\newcommand{\mem}{\mathsf{mem}}
\newcommand{\out}{\mathsf{out}}

\newcommand{\primname}{\mathrm{FESR}}
\newcommand{\PartyA}{\mathsf{A}}
\newcommand{\PartyB}{\mathsf{B}}
\newcommand{\PartyC}{\mathsf{C}}
\newcommand{\PartyAS}{\mathbf{A}}
\newcommand{\PartyBS}{\mathbf{B}}
\newcommand{\PartyCS}{\mathbf{C}}
\renewcommand{\party}{P}
\newcommand{\Store}{\textbf{store }}
\newcommand{\env}{\mathcal{Z}}

\newcommand{\sid}{\mathsf{sid}}
\newcommand{\idx}{\mathsf{idx}}
\newcommand{\pid}{\mathsf{pid}}
\newcommand{\pcode}{\mathsf{code}}

\newcommand{\fe}{\mathsf{FE}}

\newcommand{\Ffamily}{\mathtt{F}}
\newcommand{\fedomain}{\mathcal{X}}
\newcommand{\ferange}{\mathcal{Y}}
\newcommand{\ferand}{\mathcal{R}}
\newcommand{\festates}{\mathcal{S}}
\newcommand{\FfamilyDef}{\fedomain \times \festates \times \ferand \rightarrow \ferange \times \festates}
\newcommand{\fes}{\mathsf{Setup}}
\newcommand{\fekg}{\mathsf{KeyGen}}
\newcommand{\fee}{\mathsf{Enc}}
\newcommand{\fed}{\mathsf{Dec}}
\newcommand{\fein}{\mathrm{x}}
\newcommand{\feout}{\mathrm{y}}
\newcommand{\ferc}{\mathrm{r}} 
\newcommand{\funcF}{\mathrm{F}}
\newcommand{\leakfunc}{\funcF_0}

\newcommand{\festate}{\mathsf{s}}
\newcommand{\initstate}{\emptyset}

\newcommand{\handle}{\mathsf{h}}
\newcommand{\getHandle}{\texttt{getHandle}}

\newcommand{\partySet}{\mathcal{P}}
\newcommand{\preimages}{\mathcal{M}}

\newcommand{\false}{\mathsf{false}}

\newcommand{\feksg}{\mathsf{KeyShareGen}}
\newcommand{\dfesr}{DD\textrm{-}\fesr}
\newcommand{\secp}{\lambda}

\newcommand{\lenkeyshare}{k}
\newcommand{\lenkeysharef}{\lenkeyshare_\funcF}

\newcommand{\CPartyAS}{\hat{\PartyAS}}
\newcommand{\CPartyBS}{\hat{\PartyBS}}
\newcommand{\keysharelist}{\mathcal{KS}}

\newcommand{\cert}{\mathsf{CERT}}
\newcommand{\certg}{GetK}
\newcommand{\certs}{Sign}

\newcommand{\funccert}{\funct{\cert}}
\newcommand{\certout}{\mathsf{cert}}
\newcommand{\certvk}{\mathsf{VK}}

\newcommand{\Steel}{\mathsf{Steel}}
\newcommand{\dsteel}{DD\textrm{-}\Steel}
\newcommand{\mpk}{\mathsf{mpk}}
\newcommand{\msk}{\mathsf{msk}}

\newcommand{\pkd}{\mathsf{pk}_{KD}}  
\newcommand{\skd}{\mathsf{sk}_{KD}}
\newcommand{\pkf}{\mathsf{pk}_{\mathsf{FD}}}
\newcommand{\skf}{\mathsf{sk}_{\mathsf{FD}}}
\newcommand{\ct}{\mathsf{ct}}
\newcommand{\ctsk}{\mathsf{ct}_{\mathsf{key}}}
\newcommand{\ctmsg}{\mathsf{ct}_{\mathsf{msg}}}
\newcommand{\gvk}{\mathsf{vk}_{\mathsf{att}}}
\newcommand{\Flist}{\mathcal{K}}

\newcommand{\dsteelkshares}{\mathcal{KS}}

\newcommand{\aggregfunc}{\mathsf{Agg}}
\newcommand{\aggregsfunc}{\mathsf{AggS}}

\newcommand{\enplus}{\enscheme^{+}}
\newcommand{\enplusideal}{\mathcal{F}_{\enplus}}

\newcommand{\phyrecsec}{\mathsf{SEC}}

\newcommand{\enanalfuncset}{\mathit{AF}}
\newcommand{\enmuldomain}{\llbracket\mathcal{X}\rrbracket}
\newcommand{\enmulrange}{\mathcal{Y}}
\newcommand{\enanaldomain}{\enmuldomain \times \festates \times \ferand}
\newcommand{\enanalrange}{\enmulrange \times \festates}
\newcommand{\enanalfunc}{\alpha}
\newcommand{\enanalyst}{\ddot{\mathsf{A}}}
\newcommand{\enanalysts}{\overline{\mathit{A}}}
\newcommand{\enanalcorr}{\widetilde{\mathbf{\enanalyst}}}
\newcommand{\enanalstate}{\mathcal{ST}}
\newcommand{\enkeythreshold}{K}

\newcommand{\glassvault}{\mathsf{Glass}\textrm{-}\mathsf{Vault}}
\newcommand{\simglassvault}{\simul_{\mathsf{GV}}}

\newcommand{\Sim}{\mathcal{S}}

\newcommand{\simrepo}{\mathcal{H}}
\newcommand{\gattsent}{\mathcal{G}}
\newcommand{\gattsigned}{\mathcal{B}}

\newcommand{\Forward}{\textbf{forward}~}
\newcommand{\IRONsim}{\simul_{\primname}}

\newcommand{\dfesrsim}{\simul_{\dfesr}}
\newcommand{\fesrsim}{\IRONsim}
\newcommand{\dfesrsimKS}{K}
\newcommand{\dfesrsimKShares}{KS}
\newcommand{\Gprime}{\mathsf{GG}}
\newcommand{\Notify}[4]{\textbf{notify} #1 that #2 was sent from #3 to #4}
\newcommand{\NotifyAndReceive}[5]{\Notify{#1}{#2}{#3}{#4} and \textbf{capture response} #5}

\newcommand{\Await}{\textbf{await}~}

\newcommand{\ensimtbb}{\mathcal{T}}

\renewcommand{\algorithmiccomment}{\Comment}

\begin{document}
\title{Glass-Vault: A Generic Transparent Privacy-preserving Exposure Notification Analytics Platform} 
\titlerunning{Glass-Vault: An Exposure Notification Analytics Platform}
%

\author{}
\institute{}
\author{Lorenzo Martinico\inst{1}\thanks{email: lorenzo.martinico@ed.ac.uk, orcidID: 0000-0002-4968-674X}\hspace{1mm} \and \hspace{1mm}  Aydin Abadi\inst{2}\thanks{email: aydin.abadi@ucl.ac.uk, orcidID: 0000-0002-1414-8351} \and\\
Thomas Zacharias\inst{1} \thanks{email: thomas.zacharias@ed.ac.uk, 0000-0002-5022-8543}\hspace{1mm} \and \hspace{1mm} 
Thomas Win\inst{3}\thanks{email: thomas.win@uwe.ac.uk, 0000-0002-4977-0511 }}

\authorrunning{Martinico et al.}
%
\institute{University of Edinburgh \and
University College London
\and University of the West of England}
\maketitle              

%
%
%

\begin{abstract}
  The highly transmissible COVID-19 disease is a serious threat to people's
  health and life. To automate tracing those who have been in close physical
  contact with newly infected people and/or to analyse tracing-related data,  
  researchers have proposed various ad-hoc programs that require being executed
  on users' smartphones.
  Nevertheless, the existing solutions have two primary limitations: (1)
  \emph{lack of generality}: for each type of analytic task, a certain kind of data
  needs to be sent to an analyst; (2) \emph{lack of transparency}:  parties who provide data to an analyst are not necessarily infected individuals; therefore, infected individuals' data can be shared with others (e.g., the analyst) without their fine-grained and direct consent.
  In this work, we present $\glassvault$, a protocol that addresses both
  limitations simultaneously. It allows an
analyst to run authorised programs over the collected data of infectious
users, without learning the input data.
  %
  %
 $\glassvault$ relies on a new variant of generic Functional Encryption that we propose in this work. This new variant, called $\dsteel$,  offers these two additional properties: dynamic and decentralised. We illustrate the security of both $\glassvault$ and $\dsteel$ in the Universal
  Composability setting. $\glassvault$ is the first UC-secure protocol that
  allows analysing the data of Exposure Notification users in a privacy-preserving
  manner. As a sample application, we indicate how it can be used to generate ``infection heatmaps''.
\end{abstract}


\section{Introduction}

The Coronavirus (COVID-19) pandemic has been significantly affecting
individuals' personal and professional lives as well as the global economy. The
risk of COVID-19 transmission is immensely high among people in close proximity.
Tracing those individuals who have been near recently infected people (Contact
Tracing) and notifying them of a close contact with an infectious individual
(Exposure notification) is one of the vital approaches to efficiently
diminishing the spread of COVID-19 \cite{CDC}. These practices allow to identify
and instruct only those who have potentially contracted the virus to
self-isolate, without having to require an entire community to do so. This is
crucial when combating a pandemic, as widespread isolation can have destructive
effects on people's (mental and physical) well-being and countries' economies.
Researchers have proposed numerous solutions that can be installed on users'
smartphones to improve the efficiency of data collection through automation.
Most of the proposed solutions attempt to implement privacy-preserving techniques,
such as hiding infected users' contact graphs or adopting designs that prevent
tracking of non-infected users~\cite{martin20demyscoviddigitcontactracin}, to
engender sufficient adoption throughout the
population~\cite{10.1093/jamia/ocaa240}. There have also been a few ad-hoc
privacy-preserving solutions that help \emph{analyse} other user-originated data, e.g.,
location history to map the virus clusters \cite{EPRINT:BHKRW20}, or scanning QR
codes to notify those who were co-located with infected individuals
\cite{PoPETS:LGVBSPT21}.
%
%

This kind of data analytics can play a crucial role in better understanding the
spread of the virus and ultimately helps inform governments' decision-making
(e.g., to close borders, workplaces, or schools) and enhance public health
advice, especially when the analytics result is combined with the general
public's health records (e.g., a popluation's average age and common risk
factors, or people's previous exposure to infections). Despite the importance of
this type of solutions, only a few have been proposed that can preserve users'
privacy. However, they suffer from two main limitations.

Firstly, they \emph{lack generality}, in the sense
that for each type of data analysis, a certain data type/encoding has to be sent to the
analyst, which ultimately (i) limits the applications of such solutions, (ii)
increases user-side computation, communication, and storage costs, and (iii)
demands a fresh cryptographic protocol to be designed, defined, and proven secure
for every operation type. The solutions proposed in \cite{EPRINT:BHKRW20,PoPETS:LGVBSPT21} are examples of such ad-hoc analytics protocols.
It is desirable that all kinds of user data,
independent of their format could be securely collected and transmitted
through a unified protocol.
%
%
In concrete terms,
such a unified protocol will (a) provide a generic framework for scientists and health
authorities to focus on the analysis of data without having to design ad-hoc security protocols and (b) relieve users from installing
multiple applications on their devices to concurrently run data capturing
programs, which could cause issues, especially for users with resource-constrained devices.

Secondly, existing solutions \emph{lack transparency}, meaning that users are
not in control of what kind of sensitive data is being collected about them by
their contact tracing applications.
In some of these schemes, the data does not even originate directly from the
users, but from a third-party data collector, e.g., a mobile service provider
or national health service or even security services \cite{Isreal-COVID-19}.

The importance of trustworthy access to
raw health data has been recognised by the UK's
National Health Service (NHS). It has recently established a large-scale mechanism called ``Trusted Research Environments''
(TRE) that lets health data be analysed transparently and securely, by authorised
researchers \cite{ukhealthdataresearchalliance20215767586}. However, even this advanced scheme does not give individuals the possibility to explicitly choose and
withdraw consent on their data being used as part of these analytics programs, a
right that privacy legislations across the world have increasingly begun to
recognise \cite{chassang2017impact}.


\

\noindent\textbf{Our Contributions.} To address the aforementioned limitations,
in this work, we propose a platform
called ``$\glassvault$''.
$\glassvault$ is an extension of regular privacy-preserving decentralised
contact tracing, which additionally allows infected users to share sensitive
(non-contact tracing) data for analysis. It is a \emph{generic} platform, as it
is (data) type agnostic and supports \emph{any secure computation}
that users authorise. It also offers transparency and privacy, by allowing users
to consensually choose what data to share, and forcing an analyst to execute
only those computations authorised by a sufficient number of users, without
being able to learn anything about the users' inputs beyond the result. %
 %

In this work, we formally define $\glassvault$ in the Universal
Composability (UC) paradigm and prove its security. We consider a setting where
both analysts and users are semi-honest parties who may collude with each other
to learn additional information about the data beyond what the analysts are
authorised to compute. Furthermore, we allow the adversary to dynamically corrupt
users, but assume a static set of corrupted analysts.
%
%
To have an efficient protocol that can offer all the above features, we
construct $\glassvault$ by carefully combining pre-existing Exposure
Notification algorithms with an extension of the generalised functional encryption proposed
by~\citet{PKC:BKMT21}. 
We extend their construction into a novel
protocol called $\dsteel$, which allows a functional key to be generated in a
distributed fashion while letting any authorised party freely join a key
generator committee. We believe $\dsteel$ is of independent interest. As a concrete application of
$\glassvault$, we show how it can be utilised to help the analyst identify
clusters of infections through a heatmap.


The rest of the paper has been organised as follows. Section \ref{sec::background}  presents an overview of related work and provides key concepts we rely on. Section \ref{sec::DD-FESR} presents a formal definition of the generic Functional Encryption's new variant, called $\dfesr$, along with $\dsteel$, a new protocol that realises $\dfesr$. Section \ref{sec::gv} presents $\glassvault$'s formal definition, $\glassvault$ protocol, and its security proof. Section \ref{sec::heatmap} provides a concrete example of a computation that an analyst in  $\glassvault$ can perform, i.e., infection heatmaps. Appendices \ref{app::fe} and \ref{app::func} provide more detail on the underlying key concepts, while Appendix \ref{app::heat} elaborates on how the infection heatmaps can be implemented. Appendix \ref{app::steel-sim} provides more detail about the concept used in $\dfesr$ and $\dsteel$.

\section{Background}\label{sec::background}

We now give a brief overview of existing literature regarding our security
framework (\ref{subsec::background_uc}) and the building blocks for our
constructions (\ref{subsec::background_ppct}-\ref{subsec::background_fe})

\subsection{Universal Composability}\label{subsec::background_uc}

There are various paradigms via which a cryptographic protocol might be defined
and proven, such as game-based or simulation-based paradigms. Universal
Composability (UC), introduced by \citet{FOCS:Canetti01}, is a simulation-based
model that ensures security even if multiple instances of a protocol run in
parallel.
%
%
Informally, in this paradigm, the security of a protocol is shown to hold by
comparing its execution in the real world with an ideal world execution, given a
trusted \emph{ideal functionality} that precisely captures the appropriate
security requirements. A bounded environment $\env$, which provides the parties
with inputs and schedules the execution, attempts to distinguish between the two
worlds. To show the security of the protocol, there must exist an ideal world
adversary, often called a simulator, which generates a protocol's transcript
indistinguishable from the real protocol. We say \emph{the protocol UC-realises
  the functionality}, if for every possible bounded adversary in the real world
there exists a simulator such that $\env$ cannot distinguish the transcripts of
the parties' outputs. Once a protocol is defined and proven secure in the UC
model, other protocols, that use it as a subroutine, can securely replace the
protocol by calling its functionality instead, via the UC composition theorem.

Later, \citet{TCC:CDPW07} extend the original UC framework and provide a
generalised UC (GUC) to re-establish UC's original intuitive guarantee even for
protocols that use \emph{globally} available setups, such as public-key
infrastructure (PKI) or a common reference string (CRS), where all parties and
protocols are assumed to have access to some global information which is trusted
to have certain properties. GUC formalisation aims at preventing bad
interactions even with adaptively chosen protocols that use the same setup.
\citet{TCC:BCHTZ20} proposed a new UC-based framework, UC with Global
Subroutines (UCGS), that allows the presence of global functionalities in
standard UC protocols.

\subsection{Privacy-Preserving Contact Tracing}\label{subsec::background_ppct}

In this section, we present an overview of various variants of privacy-preserving solutions for contact tracing.

\subsubsection{Centralized vs Decentralized Contact Tracing.}
Within the first few months of the COVID-19 pandemic, a large number of
theoretical (e.g.,
\cite{EPRINT:ABIV20,EPRINT:ReiBraSch20c,DBLP:journals/debu/TrieuSSSS20,canetti2020anonymous})
and practical (e.g.,
\cite{EPRINT:AISEC20,EPRINT:Gvili20,DBLP:journals/access/AhmedMXRMKSHJJ20,abbas2020covid,alsdurf20coviwhitepaper})
automated contact tracing solutions were quickly developed by governments,
industry, and academic communities. Most designs concentrated around two
architectures, so-called ``centralised'' and ``decentralised''. To put it
simply, the major difference between the two architectures rests on key
generation and exposure notification. In a centralised system, the keys are
generated by a trusted health authority and distributed to contact tracing
users. In decentralised systems, the keys are generated locally by each user. In
both types, information is exchanged in a peer-to-peer fashion, commonly through
Bluetooth Low Energy (\text{BLE}) messages broadcast by each user's phone. Once
someone is notified of an infection, they upload some data to the health
authority server. While in centralised systems the uploaded data usually
corresponds to the BLE broadcast the infected users' devices listened to, in a
decentralised system it will generally be the messages they sent. Since the
(centralised) authority knows the identity of each party in the system, it can
notify them about exposure. In this setting, the authority is able to construct
users' contact graphs, which allows it to further analyse users' movements and
interactions, at the cost of privacy. On the other hand, the decentralised users
download the list of broadcasts from exposed users and compare it with their own
local lists. This guarantees additional privacy compared to centralised systems
(although several attacks are still possible, see
\cite{martin20demyscoviddigitcontactracin}), while also preventing the health
authority from running large scale analysis on population infection which would
be possible in a centralised system. Despite this, the adoption of
decentralised systems such as DP-3T \cite{DP3T} has been more widespread due to
technical restrictions and political decisions forced by smartphone
manufacturers \cite{GAEN2020}. There has been much debate on how any effort to
the adoption of a more private and featureful contact tracing scheme could be
limited by these
gatekeepers~\cite{EPRINT:Vaudenay20b,EPRINT:ABIV20,vaudenay2020dark}.

\vspace{-3mm}
\subsubsection{Formalising Exposure Notification in the UC framework.}
\citet{EPRINT:CKLRSSTVW20} introduce a comprehensive approach to formalise the Exposure
Notification primitive via the UC framework, showing how a protocol similar to
DP-3T realises their ideal functionality. Their UC formulation is designed to
capture a wide range of Exposure Notification settings. The modelling relies on a variety of
functionalities that abstract phenomena such as physical reality events and Bluetooth
communications. 
While the above work is unique in formalising Exposure notification, a UC
formalisation of the related problem of proximity testing has been given in \cite{PROVSEC:TonDavAlv11}, based on the reduction to Private Equality Testing by
\citet{NDSS:NTLHB11}.

\subsubsection{Automated Data Analysis.}
A few attempts have been made to develop automated systems which can analyse population behaviour to better understand the spread of the virus.  
%
%
%
The solution proposed in \citet{EPRINT:BHKRW20} displays the development of virus
hotspots, as a heatmap. In this system, there are two main players; namely,
the health authority and a mobile phone provider, each of which has a set of
data that they have independently collected from their users. Their goal is to
find (only) the heatmap in a privacy-preserving manner, i.e., without revealing
their input in plaintext to their counterparty. To achieve its goal, the system
uses (a computationally expensive) homomorphic encryption, differential privacy,
and the matrix representation of inputs. Thus, in this system, (i) the two
parties run the computation on users' data, without having their fine-grained
consent and (ii) each party's input has to be encoded in a specific way, i.e.,
it must be encrypted and represented as a matrix.

%

The protocols in
\cite{PoPETS:LGVBSPT21,EPRINT:BCKKMSSY20,cryptoeprint:2020:1546} allow users to
provide their encoded data to a server for a specific analysis. Specifically,
\citet{PoPETS:LGVBSPT21} design a privacy-preserving ``Presence-Tracing'' system
which notifies people who were in the same venue as an infected individual. The
proposed solution mainly uses identity-based encryption, hash function, and
authenticated encryption to achieve its goal and encode users' input data.
 %
%
%
%
%
\citet{EPRINT:BCKKMSSY20} design a privacy-preserving scheme to anonymously
collect information about a social graph of users. In this solution, a central server can 
construct an anonymous graph of interactions between users which would let the server 
understand the progression of the virus among users. This solution is based 
on zero-knowledge proofs, digital signatures, and RSA-based accumulators. 
%
%
\citet{cryptoeprint:2020:1546} propose a privacy-preserving scheme between
multiple non-colluding servers to help epidemiologists simulate and predict
future developments of the virus. This scheme relies on heavy machineries such
as oblivious shuffling, anonymous credentials, and generic multi-party
computation.

In all of the above solutions, the parties need to encode their inputs
in a certain way to support the specific computation that is executed on their
inputs, and thus do not support generality. Additionally, not all systems allow
the users to opt-out of the computation, and are therefore not transparent.
%

\vspace{-4mm}
\subsection{Trusted Execution Environments}\label{subsec::background_tee}
Trusted Execution Environments (TEEs) are secure processing environments (a.k.a.
{secure enclaves}) that consist of processing, memory, and storage hardware
units. In these environments, the residing code and data are isolated from other
layers in the software stack including the operating system. TEEs ensure that
the integrity and confidentiality of data residing in them are preserved. They
can also offer remote attestation capabilities, which enable a party to remotely
verify that an enclave is running on a genuine TEE hardware platform. Under the
assumption that the physical CPU is not breached, enclaves are protected from an
attacker with physical access to the machine, including the memory and the
system bus. The first formal definition of TEEs notion was proposed by
\citet{EC:PasShiTra17} in the GUC model. This definition provides a basic
abstraction called the \emph{global attestation functionality} ($\Gatt$,
described in Supporting material~\ref{app::func_Gatt}), that casts the core services a wide
class of real-world attested execution processors offer. $\Gatt$ has been used
by various protocols both in the GUC model (e.g., in
\cite{EPRINT:TZLHJS16,ASIACCS:WSDLZW19}) and recently in the UCGS model (in
\cite{PKC:BKMT21}). 
TEEs have been used in a few protocols to implement Contact Tracing solutions, e.g., in
\cite{castelluccia20_desir,sturzenegger20_confid_comput_privac_preser_contac_tracin,mailthody21_safer_illin_rokwal}.
%
%

\vspace{-4mm}
\subsection{Functional Encryption}\label{subsec::background_fe}

Informally, ``Functional Encryption'' ($\fe$) is an encryption
 scheme that ensures that a party who possesses a decryption key learns nothing beyond a specific function
  of the encrypted data. Many types of encryption (e.g., identity-based or attribute-based encryptions)
  can be considered as a special case of $\fe$. The notion of $\fe$ was first defined formally by  \citet{TCC:BonSahWat11}. $\fe$ involves
  three sets of parties: $\PartyAS, \PartyBS$, and $\PartyCS$. Parties in
  $\PartyAS$ are encryptors, parties in $\PartyBS$ are decryptors, and parties
  in $\PartyCS$ are key generation authorities. The syntax of $\fe$ is recapped
  in Supporting material~\ref{app::fe}.


 %
 An $\fe$ scheme can also be
thought of as a way to implement access control to an ideal repository
\cite{EPRINT:MatMau13}.
%
Since the introduction of $\fe$, various variants of $\fe$ schemes have been
proposed, such as those that (i) support distributed ciphertexts, letting joint
functions be run on multiple inputs of different parties; (ii) support
distributed secret keys that do not require a single entity to hold a master
key; or (iii) support both properties simultaneously.
  %
  %
   %
In particular (ii) and (iii) affect the membership of sets
$\PartyAS$ and $\PartyCS$. We refer readers to \cite{abs-2106-06306,AgrawalGT21}
for surveys of $\fe$ schemes.
Recently, \citet{PKC:BKMT21} proposed $\fesr$, a generalisation of $\fe$ that
allows the decryption of the class of \emph{Stateful} and \emph{Randomised}
functions. This class is formally defined as
$\Ffamily=\{\funcF\mid \funcF : \FfamilyDef\}$, where $\festates$ is the set of
possible function states and $\ferand$ is the universe of random coins. They
also provide a protocol, called $\Steel$, that realises $\fesr$ in the UCGS
model and relies on TEEs (as abstracted by $\Gatt$). By introducing functions
with state and randomness, $\fesr$ allows computing a broader class of functions
than most other $\fe$ schemes. Moreover, other appealing properties of the
scheme include its secure realisation in the UC paradigm, and that it can be easily
extended to support multiple inputs (as we show in Subsection~\ref{sec::mi-fe}).
\label{sec:PPCT}



\section{Dynamic and Decentralised FESR ($\dfesr$) and Steel ($\dsteel$)}\label{sec::DD-FESR}



In this section, we present the ideal functionality $\dfesr$ and the protocol that realises it,
$\dsteel$. As we stated earlier, they are built upon the original functionality
$\fesr$ and protocol $\Steel$, respectively. In the formal descriptions, we will highlight the main changes that we have applied to the original scheme (in \cite{PKC:BKMT21}) in yellow. For conciseness, we omit any reference to UC-specific
machinery such as ``session ids'' unless it is required to understand the specifics
of our protocol.

\subsection{The ideal functionality $\dfesr$}\label{subsec::DD-FESR}
In this subsection, we extend $\fesr$ into a new functionality $\dfesr$
to capture two additional properties from the functional encryption
literature:

\begin{itemize}
  \item \emph{Decentralisation} (introduced in \cite{AC:CDGPP18}) allows a set
        of encryptors to be in control of functional key generation, rather than
        a single trusted authority of type $\PartyCS$. The \textsc{KeyGen}
        subroutine of the $\fesr$ scheme is replaced by a new subroutine
        \textsc{KeyShareGen}, which can be run by any party $\PartyA\in\PartyAS$
        to produce a ``key share''. In this case, a decryptor $\PartyB$ needs to
        collect at least $\lenkeyshare$ shares for some function $\funcF$ before
        $\PartyB$ is allowed to decrypt. This threshold parameter,
        $\lenkeyshare$, is specified by $\PartyA$ when it wants to encrypt, and
        is unique for each ciphertext, but does not restrict key share
        generation to a specific subset of $\PartyAS$.

  \item \emph{Dynamic} membership (introduced in \cite{C:CDGPP20}) allows any party
        to freely join set $\PartyAS$ during the execution of the protocol. In
        our instantiation, a new party $\PartyA$ joins through a local procedure
        which only requires the public parameters. $\dfesr$ can be instantiated
        with some preexisting $\PartyAS$ members, or with an empty set that
        is gradually filled through $\fes$ calls. Our current scheme is
        permissionless, meaning anyone can register as a new party.
\end{itemize}

\vspace{1mm}

\begin{functionality}{$\dfesr[\Ffamily,\PartyAS,\PartyBS,\PartyC$]}


  \noindent The functionality is parameterised by the randomized function class
  $\Ffamily=\{\funcF\mid \funcF: \FfamilyDef\}$, over state space $\festates$ and
  randomness space $\ferand$, and by three distinct types of party entities
  $\PartyA\in\PartyAS,\PartyB\in\PartyBS,\PartyC$ interacting with the
  functionality via dummy parties (that identify a particular role).
  %
  %
  \begin{statedecl} 
  \statevar{\CPartyAS}{[]}{List of corrupted As}
  \statevar{\CPartyBS}{[]}{List of corrupted Bs}
  \statevar{\leakfunc}{}{Leakage function returning the length of the message}
  \statevar{\Ffamily^+}{}{Union of the allowable functions and leakage function, i.e., $\Ffamily \cup \leakfunc$}
	\statevar{\mathsf{setup[\cdot]}}{\false}{ Table recording which parties were initialized.}
  \statevar{\preimages[\cdot]}{\bot}{Table storing the plaintext for each message handler}
  \statevar{\partySet[\cdot]}{\initstate}{Table of authorised functions' states for all decryption parties}
  \statevar{\keysharelist[\cdot]}{[]}{Table of key share generator  for each  (decryptor, function) pair}
  \end{statedecl}


  \begin{receive}[\party]{setup}{}  
    \State \Assert $\mathsf{setup}[\party] = \false$
    \State \AsyncSendAndReceive{setup}{\party}{\adversary}{$\hl{OK}$}
    \If{\hl{$\party.\pcode = \PartyA$}}
    \hl{$\PartyAS \gets \PartyAS \concat \party$} \algorithmiccomment{{\scriptsize{Party
      membership is determined by the code in the UC identity tape}}}
    \ElsIf{\hl{$\party.\pcode = \PartyB$}} \hl{$\PartyBS \gets \PartyBS \concat \party$}
    \Else{} \Return \EndIf
    \State$\mathsf{setup}[\party] \gets\mathsf{true}$
  \end{receive}
  \begin{receive}[\PartyA \in \PartyAS]{\hl{KeyShareGen}}{$\hl{$\funcF,\PartyB$}$}   
    \If{\hl{$\PartyA \in \CPartyAS$}} \algorithmiccomment{{\scriptsize{The adversary can only block key generation for
       corrupted parties}}}
     \State \AsyncSendAndReceive{\hl{keysharequery}}{$\hl{$\funcF,\PartyA,\PartyB$}$}{\adversary}{\msg{OK}{}}
     \EndIf
    \If{$\big(\funcF \in \Ffamily^{+}\land\mathsf{setup[\PartyA]} \land\mathsf{setup[\PartyB]} \big)$}
    \State
    \hl{$\keysharelist[\PartyB,\funcF] \gets \keysharelist[\PartyB,\funcF] \concat \PartyA$}\algorithmiccomment{{\scriptsize{We
    store the identity of all parties in $\PartyAS$ who authorised $\funcF$ for $\PartyB$}}}
    \State \Send{\hl{keysharegen}}{$\hl{$\funcF, \PartyA,\PartyB$}$}{\PartyB}
    \If{\hl{$\PartyB \in \CPartyBS \lor \PartyA\in\CPartyAS$}}{} \Send{\hl{keysharegen}}{$\hl{$\funcF, \PartyA,\PartyB$}$}{\adversary}
    \Else{} \Send{\hl{keysharegen}}{$\hl{$\bot, \PartyA,\PartyB$}$}{\adversary}
    \EndIf
    \EndIf
  \end{receive}
	\begin{receive}[\text{party } \party \in \{\PartyAS \cup \PartyBS\}]{Encrypt}{\fein,$\hl{$\lenkeyshare$}$}
      \If{$\big(\mathsf{setup}[\party] \land \fein \in \fedomain\land$\hl{$\lenkeyshare\mbox{ is an integer}$}$\big)$}
      \State compute $\handle \gets \getHandle$ \algorithmiccomment{{\scriptsize{Generate a unique index, $\handle$, by running the subroutine $\getHandle$}}}
      \State $\preimages[\handle] \gets (\fein,$\hl{$\lenkeyshare$})
      \State\Send{encrypted}{\handle}{\party}
    \Else
      \State\Send{encrypted}{\bot}{\party}
    \EndIf
  \end{receive}
  \begin{receive}[\text{party } \PartyB \in \PartyBS]{Decrypt}{\funcF,\handle}
    \State $(\fein,$\hl{$\lenkeyshare$}$) \gets \preimages[\handle]; \feout\gets\bot; $
    \If{$\funcF=\leakfunc$}
    \State $\feout\gets|\fein|$
   
    \ElsIf{\hl{$\big(( |\keysharelist[\PartyB,\funcF]| \geq \lenkeyshare \land \forall \PartyA\in\keysharelist[\PartyB,\funcF] : \mathsf{setup}[\PartyA] \land \forall\PartyA$
      are distinct $ )\lor$}\\\hl{$\quad\quad\lor (\PartyB\in\CPartyBS\land |\CPartyAS| \geq \lenkeyshare)\big)$}}
     \algorithmiccomment{{\scriptsize{There are at least $\lenkeyshare$ functional key shares, all
        generated by correctly setup\\  parties, OR $\PartyB$ and at least $\lenkeyshare$ parties in $\PartyAS$ are corrupted }}}
        \State $\festate \gets \partySet[\PartyB,\funcF]$
        \State $\ferc \samuni \ferand$
        \State $(\feout,\festate') \gets \funcF(\fein,\festate;\ferc)$
        \State $\partySet[\PartyB,\funcF] \gets \festate'$
    \EndIf
    \State \Return $\msg{decrypted}{\feout}$
  \end{receive}
  \begin{receive}[\adversary]{Corrupt}{\party}
  \If{\hl{$\party\in\PartyAS$}}
   \State\hl{$\CPartyAS \gets \CPartyAS \concat \party$}\algorithmiccomment{{\scriptsize{ The
   functionality needs to keep track of corrupted parties in $\PartyAS$ to ensure\\ correctness}}}
   \State\Return\hl{$\{(\PartyB,\funcF) | \party\in \keysharelist[\PartyB,\funcF]\}$}
  \EndIf 
   \If{\hl{$\party\in\PartyBS$}}
   \State \hl{$\CPartyBS \gets \CPartyBS \concat \party$}
   \State\Return\hl{$\keysharelist[\party,\cdot]$}
  \EndIf
  \end{receive}
\end{functionality}

\subsection{The protocol $\dsteel$}\label{subsec::DD-Steel}
Now, we propose $\dsteel$, a new protocol that extends the original $\Steel$ to
realise $\dfesr$. We first briefly provide an overview of $\Steel$ and our new
extension, before outlining the formal protocol.
$\Steel$ uses a public key encryption scheme, where the master public
key is distributed to parties in $\PartyAS$, and the master secret key is
securely stored in an enclave running program $\prog{\KME}$ on trusted party
$\PartyC$. The secret key is then provisioned to enclaves running on parties in
$\PartyB$, only if they can prove through remote attestation that they are
running a copy of $\prog{\DE}$. The functional key corresponds to signatures over
the representation of a function $\funcF$ and is generated by $\PartyC$'s $\prog{\KME}$
enclave. If party $\PartyB$ possesses any such key, its copy of $\prog{\DE}$
will distribute the master secret key to a $\prog{FE[\funcF]}$ enclave, which
can then decrypt any $\PartyA$'s encrypted inputs $\fein$ and will ensure that only
 value $\feout$  of $(\feout,\festate') \gets \funcF(\fein,\festate;\ferc)$ is
returned to $\PartyB$, with the function states $\festate$ and $\festate'$ protected by
the enclave.

The protocol $\dsteel$ is similar to the original version as described, except
for a few crucial differences. 
Party $\PartyC$, who is now untrusted, still runs the public key encryption
parameter generation within an enclave and distributes it to a party $\PartyA$
or $\PartyB$ when they first join the protocol. Each party $\PartyA$ also
generates a digital signature key pair locally and includes a key policy
$\lenkeyshare$ along with their ciphertext. Note that for simplicity our current
version uses an integer $\lenkeyshare$ to associate with each message, but it
would be possible to use a public key policy as a threshold version of
Multi-Client Functional Encryption. 
Party $\PartyA$ who wants to authorise party $\PartyB$ to compute a certain function
will run $\feksg(\funcF,\PartyB)$ to generate a key share, which requires signing the
representation of $\funcF$ with their local key, and send the signature to
$\PartyB$. 
For $\PartyB$'s $\prog{\DE^{\certvk}}$ enclave to authorise the functional decryption of $\funcF$, it
first verifies that all key shares provided by $\PartyB$ for $\funcF$ are valid
and each was provided by a unique party in $\PartyAS$; if all checks are passed, 
 then it will distribute the master secret key to $\prog{\FE^{\certvk}[\funcF]}$,
along with the length of recorded key shares $\lenkeysharef$.
$\prog{\FE^{\certvk}[\funcF]}$ will only proceed with decryption if the number of
key shares meets the encryptor's key policy.
Provisioning of the secret key between $\PartyC$ and $\PartyB$'s $\prog{\DE^{\certvk}}$
enclave remains as in $\fesr$.

The protocol $\dsteel$ makes use of the global attestation functionality
$\Gatt$, the certification functionality $\funccert$, the common reference
string functionality $\CRS$, the secure channel functionality $\SC_R^S$, and the
repository functionality $\REP$ that are presented in Supporting material~\ref{app::func_Gatt}, \ref{app::func_Fcert}, \ref{app:func_crs}, \ref{app::func_sc}, and~\ref{app::func_rep}  respectively. The code of the enclave programs
$\prog{\KME^{\certvk}},\prog{\DE^{\certvk}},\prog{\FE^{\certvk}[\cdot]}$ in $\dsteel$ is
hardcoded with the value of the verification key $\certvk$ returned by $\funccert$, and can be generated during the protocol
runtime.

\vspace{4mm}

\begin{protocol}{$\dsteel[\Ffamily,\pke,\sigscheme,\nizkscheme,\secp]$}

\noindent The protocol is parameterised by the class of functions $\Ffamily$ as defined in
$\dfesr$, the public-key encryption scheme $\pke$ denoted as the triple of algorithms $\pke:=(\pkeg,$ $\pkee,\pked)$, the digital signature scheme $\sigscheme$ denoted as the triple of algorithms $\sigscheme:=(\sigg,\sigs,\sigv)$, the non-interactive zero-knowledge protocol $\nizkscheme$ that consists of prover $\Prv$ and verifier $\Ver$, and the security parameter $\secp$.

\begin{statedecl}
 \statevar{\dsteelkshares[\cdot]}{\emptyset}{Table of function key shares at $\PartyB$}
 \statevar{\Flist[\cdot]}{\emptyset}{Table of functional enclave details at $\PartyB$}
\end{statedecl}

\noindent\textbf{\underline{Key Generation Authority $\PartyC$}:}
\begin{receive}[\SC^\party]{Setup}{\party}
  \If{$\mpk = \bot$}
    \State \SendAndReceive{Get}{}{\CRS}{\msg{Crs}{\ncrs}}
    \State \SendAndReceive{\hl{GetK}}{}{$\hl{$\funccert$}$}{$\hl{$\certvk$}$}
    \State  $\eid{\KME} \gets \Gatt.\install(\PartyC.\sid,$ \hl{$\prog{\KME^{\certvk}}$})
    \State  $(\mpk, $\hl{$\irsig{\KME}$}$) \gets \Gatt.\resume(\eid{\KME}, (\initkw,\ncrs,\PartyC.\sid))$
  \EndIf
  \If{$\party.\pcode = \PartyA$}
    \State  \Send{Setup}{\mpk,$\hl{$\irsig{\KME}$}$,$\hl{$\eid{\KME}$}$}{\SC_{\party}}
  \ElsIf{$\party.\pcode = \PartyB$}
	  \State \SendAndReceive{Setup}{\mpk,$\hl{$\irsig{\KME}$}$,\eid{\KME}}{\SC_{\party}}{\msg{provision}{\irsig{\DE}, \eid{\DE}, \pkd}}
	  \State  $(\ctsk, \irsig{sk}) \gets \Gatt.\resume(\eid{\KME}, (\provisionkw,
    (\irsig{DE}, \eid{\DE}, \pkd, \eid{\KME})))$
	  \State  \Send{provision}{\ctsk, \irsig{sk}}{\SC_{\party}}
  \EndIf
\end{receive}
\textbf{\underline{Encryption Party $\PartyA$}:}
\begin{receive}{Setup}{}
  \State \Assert $\mpk=\bot$
  \State \SendAndReceive{Setup}{\PartyA.\pid}{\SC^\PartyC}{\mpk,$\hl{$\irsig{\KME}$}$,$\hl{$\eid{\KME}$}$}
  \State \SendAndReceive{getpk}{}{\Gatt}{\gvk}
  \State \SendAndReceive{\hl{GetK}}{}{$\hl{$\funccert$}$}{$\hl{$\certvk$}$}
  \State \Assert \hl{$\sigscheme.\sigv(\gvk,(\sid,\eid{\KME},\prog{\KME^{\certvk}},\mpk),\irsig{KME})$}
  \State \SendAndReceive{Get}{}{\CRS}{\msg{Crs}{\ncrs}}
  \State \hl{$(\sigpk,\sigsk) \gets \sigscheme.\sigg(1^{\secp})$}
  \State \SendAndReceive{\hl{\certs}}{$\hl{$\sigpk$}$}{\funccert}{\certout}
  \State $\Store \mpk,\ncrs,\sigpk,\sigsk,\certout$
\end{receive}
\begin{receive}{\hl{KeyShareGen}}{$\hl{$\funcF,\PartyB$}$}
\State \hl{$\irsig{} \gets \sigscheme.\sigs(\sigsk,\funcF,\PartyB)$}
\State \Send{\hl{KeyShareGen}}{$\hl{$\funcF,\irsig{},\sigpk,\certout$}$}{\SC_B}
\end{receive}
\begin{receive}{Encrypt}{\pkemsg,$\hl{$\lenkeyshare$}$}
   \State \Assert{$\mpk \neq \bot \land \pkemsg \in \mathcal{X}\land\lenkeyshare\mbox{ is an integer} $}
   \State $\ct \xleftarrow{ \pkerc} \pke.\pkee(\mpk, (\pkemsg,$\hl{$\lenkeyshare$}$))$
   \State $\pi \gets \Prv((\mpk,\pkect),((\pkemsg,$\hl{$\lenkeyshare$}$),\pkerc),\ncrs);\ctmsg \gets (\pkect,\pi)$
   \State \SendAndReceive{write}{\ctmsg}{\REP}{\handle}
   \State \Return $\msg{encrypted}{\handle}$
\end{receive}
\textbf{\underline{Decryption Party $\PartyB$}:}
\begin{receive}{Setup}{}
  \State \Assert $\mpk=\bot$
  \State \SendAndReceive{Setup}{}{\SC^\PartyC}{\mpk,$\hl{$\irsig{\KME}$}$,\eid{\KME}}
  \State $\dsteelkshares = \{\},\Flist = \{\}$
  \State \SendAndReceive{getpk}{}{\Gatt}{\gvk}
  \State \SendAndReceive{\hl{GetK}}{}{$\hl{$\funccert$}$}{$\hl{$\certvk$}$}
  \State \Assert \hl{$\sigscheme.\sigv(\gvk,(\idx,\eid{\KME},\prog{\KME^{\certvk}},\mpk),\irsig{\KME})$}
  \State \Store $\mpk; \eid{\DE} \gets \Gatt.\install(\PartyB.\sid,$ \hl{$\prog{\DE^{\certvk}}$})
   \State \SendAndReceive{Get}{}{\CRS}{\msg{Crs}{\ncrs}}
	\State  $((\pkd, \cdot, \cdot), \irsig{}) \gets \Gatt.\resume(\eid{\DE}, (\insetupkw, \eid{\KME},\irsig{\KME},\ncrs,\PartyB.\sid))$
	\State \SendAndReceive{provision}{\irsig{}, \eid{\DE}, \pkd}{\SC_\PartyC}{\msg{provision}{\ctsk, \irsig{\KME}}}
  \State $\Gatt.\resume(\eid{\DE}, (\cosetupkw, \ctsk, \irsig{\KME}))$
\end{receive}
\begin{receive}[\SC^\PartyA]{\hl{KeyShareGen}}{$\hl{$\funcF,\irsig{},\sigpk,\certout$}$}
  \State \hl{$\dsteelkshares[\funcF] \gets \dsteelkshares[\funcF] \concat (\irsig{},\sigpk,\certout)$}
\end{receive}
\begin{receive}{Decrypt}{\funcF, \handle}
  \If{\hl{$\Flist[\funcF] = \bot$}}
  \State  $\eid{\funcF} \gets \Gatt.\install(\PartyB.\sid,$ \hl{$\prog{\FE^{\certvk}[\funcF]}$})
  \State  $(\pkf, \irsig{\funcF}) \gets \Gatt.\resume(\eid{\funcF}, (\initkw,\mpk,\PartyB.\sid))$
  \State $\Flist[\funcF] \gets (\eid{\funcF}, \pkf, \irsig{F})$
  \EndIf
  \State \SendAndReceive{read}{\handle}{\REP}{\ctmsg}
  \State $(\eid{F}, \pkf, \irsig{\funcF}) \gets \Flist[\funcF]$
  \State $((\ctsk,$\hl{$\lenkeysharef$}$, \ncrs), \irsig{\DE}) \gets \Gatt.\resume(\eid{\DE}, (\provisionkw, \dsteelkshares[\funcF],\eid{F},\pkf,\irsig{\funcF}, \funcF,\PartyB.\pid))$
  \State  $((\compkw, \feout), \cdot) \gets\Gatt.\resume(\eid{\funcF}, (\runkw, \irsig{\DE}, \eid{\DE}, \ctsk, \ctmsg,$\hl{$\lenkeysharef$}$,\ncrs,\bot))$
  \State \Return $\msg{decrypted}{\feout}$
\end{receive}

\noindent\underline{$\prog{KME^{\certvk}}$} \\
\text{on input init} \\
\pcind $(\pkepk,\pkesk) \gets \pke.\pkeg(1^{\secp})$ \\
\pcind \pcreturn $\pkepk$ \\
\text{on input ($\provisionkw, (\irsig{\DE}, \eid{\DE}, \pkd, \eid{\KME})$):} \\
   \pcind $\gvk \gets \Gatt.\gvk$; \textbf{fetch} $\ncrs,\idx,\pkesk$ \\
   \pcind \Assert $\sigscheme.\sigv(\gvk, (\idx, \eid{\DE}, \prog{\DE^{\certvk}}, (\pkd,\eid{\KME}, \ncrs), \irsig{\DE})$ \\
   \pcind $\ctsk \gets \pke.\pkee(\pkd, \pkesk)$ \\
  \pcind \pcreturn $\ctsk$\\

\noindent\underline{$\prog{DE^{\certvk}}$} \\
  \text{on input ($\insetupkw, \eid{\KME}, \ncrs,\idx$):} \\
  \pcind $\Assert \pkd\neq \bot$  \\
  \pcind $(\pkd, \skd) \gets \pke.\sigg(1^{\secp})$ \\
  \pcind $\Store  \skd,\eid{\KME},\ncrs,\idx$ \\
  \pcind \pcreturn $\pkd, \eid{\KME}, \ncrs$ \\
  \text{on input ($\cosetupkw, \ctsk, \irsig{\KME}$):} \\
  \pcind  $\gvk \gets \Gatt.\gattmpk$ \\
  \pcind \textbf{fetch } $\eid{\KME},\skd,\idx$ \\
  \pcind $\pkemsg \gets (\idx, \eid{\KME}, \prog{\KME^{\certvk}}, \ctsk)$\\
  \pcind \Assert $\sigscheme.\sigv(\gvk, \pkemsg, \irsig{\KME})$ \\
  \pcind $\pkesk \gets \pke.\pked(\skd, \ctsk)$ \\
  \pcind $\Store  \pkesk, \gvk$ \\
	\text{on input ($\provisionkw, \dsteelkshares_{\funcF}, \eid{\funcF},\pkf, \irsig{\funcF},\funcF,\pid$):} \\
  \pcind \textbf{fetch } $\eid{\KME},\gvk,\pkesk,\idx,\ncrs$ \\
  \pcind $\pkemsg \gets (\idx, \eid{}, \prog{\FE^{\certvk}[\funcF]}, \pkf)$\\
  \pcind \Assert \hl{$\forall (\irsig{\sigpk},\sigpk,\certout)\in\dsteelkshares_{\funcF}:$} \\
  \pcind \hl{$:(\sigscheme.\sigv(\certvk,\sigpk,\certout) \land$
   $ \sigscheme.\sigv(\sigpk,(\funcF,\pid),\irsig{\sigpk})\land$}\\
  \pcind \hl{$\land \forall\sigpk$ are distinct)} $\land\sigscheme.\sigv(\gvk, \pkemsg, \irsig{\funcF})$ \\
  \pcind \pcreturn $\pke.\pkee(\pkf, \pkesk),$\hl{$|\dsteelkshares_{\funcF}|$}$,\ncrs$\\
  

\noindent\underline{$\prog{\FE^{\certvk}[\funcF]}$} \\
  \text{on input ($\initkw,\mpk,\idx$):} \\
  \pcind $\Assert \pkf = \bot$ \\
  \pcind $(\pkf,\skf) = \pke.\sigg(1^{\secp}) $ \\
  \pcind $\mem \gets \emptyset; \Store  \skf,\mem,\mpk,\idx$ \\
  \pcind \pcreturn $\pkf$ \\
  \text{on input ($\runkw, \irsig{\DE}, \eid{\DE},\ctsk, \ctmsg,$\hl{$\lenkeysharef$}$, \ncrs, \feout')$:} \\
  \pcind \pcif $\feout' \neq \bot$ \\
  \pcind \pcind \pcreturn $(\compkw, \feout')$ \\
  \pcind $ \gvk \gets \Gatt.\gattmpk;  (\pkect,\pi) \gets \ctmsg$ \\
  \pcind \textbf{fetch } $\skf,\mem,\mpk,\idx$ \\
  \pcind $\pkemsg \gets (\idx, \eid{\DE}, \prog{\DE^{\certvk}}, (\ctsk,$\hl{$\lenkeysharef$}$,\ncrs))$ \\
  \pcind \Assert $\sigscheme.\sigv(\gvk, \pkemsg, \irsig{\DE}) $\\
  \pcind $\pkesk = \pke.\pked(\skf, \ctsk)$ \\
  \pcind \Assert $\nizkscheme.\Ver((\mpk,\pkect), \pi, \ncrs)$ \\
  \pcind $(\fein,$\hl{$\lenkeyshare$}$) = \pke.\pked(\pkesk, \pkect)$ \\
  \pcind \Assert \hl{$\lenkeysharef \geq \lenkeyshare$} \\
  \pcind $\out,\mem' = \funcF(\fein,\mem)$ \\
  \pcind $\Store \mem \gets \mem'$ \\
  \pcind \pcreturn $(\compkw, \out)$

\end{protocol}

\vspace{2mm}

\subsection{Proof of security}
We now formally state the security guarantees of $\dsteel$ as a Theorem:

\begin{theorem}\label{thm::ddfesr}
  For a class of functions $\Ffamily$, CCA-secure encryption scheme $\pke$, EU-CMA secure signature scheme $\sigscheme$, and non-interactive zero-knowledge proof system $\nizkscheme$,
  Protocol $\dsteel[\Ffamily,\pke,\sigscheme,\nizkscheme,\secp]$ UC-realises ideal functionality\\ $\dfesr[\Ffamily,\PartyAS,\PartyBS,\PartyCS]$, in the presence of global functionality $\Gatt$.
\end{theorem}

\begin{proof} We first construct a simulator, $\dfesrsim$. For simplicity of
  exposition, we use the simulator $\fesrsim$ in \cite{PKC:BKMT21} as a
  subroutine to $\dfesrsim$. We instantiate $\fesrsim$ such that shared
  functionalities (e.g., the secure channels between multiple parties) are
  implemented by $\dfesrsim$, so that it can intercept messages to the parties
  whose behaviour is simulating and act accordingly. $\dfesrsim$ also acts as
  the ideal functionality in the eyes of the $\fesrsim$ simulator. We reproduce
  the original $\fesrsim$, with appropriate modifications to conform to the
  message syntax of $\dfesr$ and $\dsteel$, in Supporting
  material~\ref{app:steelsim}.

  We assume that at least one party in the set $\PartyBS \cup \PartyC$ is
  corrupted at the start of the protocol. We choose this party, $\Gprime$, to
  install all $\Gatt$ enclaves for all participants in the protocol, be they
  honest or corrupted. Due to the property of anonymous attestation guaranteed
  by $\Gatt$, the simulator can install all programs on the same machine to
  produce the attested trace of the real-world protocol, as long as it does not
  allow $\Gprime$ to execute other parties' enclaves on its own initiative (we
  use the table $\gattsent$ from simulator $\fesrsim$ to keep track of which
  party installed each enclave). Similar to the original simulator, we use
  the shorthand $\outp\gets~\Gatt.\mathsf{command}(\inp)$ to indicate
  ``\SendAndReceive*[$\Gprime$][sim]{$\msg{command}{\inp}$}{$\Gatt$}{$\outp$}''; note, 
  in \cite{PKC:BKMT21}, the message was always sent from $\PartyB$ instead,
  given the simpler setting of one honest $\PartyC$ and one corrupted $\PartyB$
  in their proof.

  We also give $\dfesrsim$ white-box access to $\fesrsim$, letting the
  former freely access the internal tapes of the latter. We mark the names of
  the variables that are read from $\fesrsim$ in the state variable declaration
  below. In particular, we use the internal values of $\fesrsim$ to keep track
  of messages sent to the enclaves and their attestation signatures. For all
  calls to an enclave, the $\dfesrsim$ simulator always activates
  $\fesrsim$ so that these internal variables can be updated.

   \clearpage

  \begin{simulator}{$\dfesrsim$}
     \vspace{-7mm}
    \begin{statedecl}
    
      \statevar{\dfesrsimKS}{\{\}}{set of $\PartyAS$ keypairs and $\funccert$ certificates}
      \statevar{\dfesrsimKShares}{\{\}}{set of generated keyshares}
      \statevar{\Gprime}{\bot}{The corrupted party on which we run simulated enclaves} 
      \statevar{\gattsent = \fesrsim.\gattsent}{}{Collects all messages sent to $\Gatt$ and its response} 
      \statevar{\gattsigned = \fesrsim.\gattsigned}{}{Collects all messages signed by $\Gatt$}
      \statevar{\ncrs = \fesrsim.\ncrs}{}{Simulated common reference string}
    \end{statedecl}
    \begin{receive}[\dfesr]{Setup}{\party}

      \If{$\party.\pcode=\PartyA$}

      \If{$\party$ is honest}
        \If{$\PartyC$ is honest}
          \State \Send{setup}{\party}{\fesrsim}
        \Else
          \State \SendAndReceive{Setup}{\PartyA.\pid}{\SC^\PartyC}{\mpk,\irsig{\KME},\eid{\KME}}
          \State \Assert $(\PartyC.\sid,\eid{\KME},\prog{\KME^{\certvk}},\mpk)\in\gattsigned[\irsig{KME}]$
        \EndIf
        \State $\gvk\gets\Gatt.\getpk$
        \State \SendAndReceive{\certg}{}{\funccert}{\certvk}
        \State $(\sigpk,\sigsk) \gets \sigscheme.\sigg(1^{\secp})$
        \State \SendAndReceive*[\party][sim]{$\msg{\certs}{\sigpk}$}{$\funccert$}{$\certout$}
        \State $\dfesrsimKS[\party] \gets (\sigpk,\sigsk,\certout)$
        \State \Send{OK}{}{\dfesr}
      \Else
        \If{$\PartyC$ is honest}
          \State \Send*[$\environment$][sim]{$\msg{setup}{\party}$}{$\party$}
          \State \Await for message $\msg{setup}{\party}$ on $\SC^{\PartyC}_{\party}$
          \State \Send{setup}{\party}{\fesrsim} \State \Send{OK}{}{\dfesr}
        \Else
          \State \Send*[$\environment$][sim]{$\msg{setup}{\party}$}{$\party$}
          \State await for message $\msg{Sign}{\sigpk}$ from $\party$ to $\funccert$
          \State \Send{OK}{}{\dfesr}
        \EndIf
      \EndIf
    \ElsIf{$\party.\pcode=\PartyB$}
      \If{$\party$ is honest}
        \If{$\PartyC$ is honest}
          \State \NotifyAndReceive{$\fesrsim$}{$\msg{install}{\prog{\DE^{\certvk}}}$}{$\party$}{$\Gatt$}{$\eid{\DE}$}
          \State \SendAndReceive{setup}{\party}{\fesrsim}{\mpk, \irsig{\KME}, \eid{\KME}}
          \State \NotifyAndReceive{$\fesrsim$}{$\msg{resume}{\insetupkw,\eid{\KME},\ncrs,\party.\sid}$}{$\party$}{$\Gatt$}{$(\pkd, \eid{\KME}, \ncrs),\irsig{init}$}
          \State
          \SendAndReceive{provision}{\irsig{init},\eid{\DE},\pkd}{\fesrsim}{\msg{provision}{\ctsk,\irsig{\KME}}}

          \State \Notify{$\fesrsim$}{$\msg{resume}{\cosetupkw,\ctsk,\irsig{\KME}}$}{$\party$}{$\Gatt$}
          \State \Send{OK}{}{\dfesr} \Else
          \State \Send*[$\environment$][sim]{$\msg{setup}{\party}$}{$\party$}
          \State \Await for $\msg{setup}{}$ message on $\SC^{\PartyC}_{\party}$
          \State \Send{setup}{\party}{\fesrsim} \EndIf \Else \If{$\PartyC$ is honest}
          \State \Send{setup}{\party}{\fesrsim}
          \State \Send{OK}{}{\dfesr} \Else
          \State \Send*[$\environment$][sim]{$\msg{setup}{\party}$}{$\party$}
          \State \Await for $\msg{resume}{\cosetupkw,\ctsk,\irsig{\KME}}$ from
          $\party$ to $\Gatt$
          \State $(\idx,\eid,\prog{},\ctsk)\gets\gattsigned[\irsig{\KME}]$
          \State \Assert
          $\idx=\party.\sid \land \prog{} = \prog{\DE^{\certvk}} \land (\irsig{\KME},\cdot,\ctsk)\in\gattsent[\eid{}].\resume$

          \State \Send{OK}{}{\dfesr}
        \EndIf
      \EndIf

      \EndIf

    \end{receive}
    \begin{receive}[\text{ corrupted party } \party \text{ to } \funccert]{Sign}{\sigpk}
      \State \Forward$\msg{Sign}{\sigpk}$ and receive response $\certout$
      \State $\dfesrsimKS[\party]\gets(\sigpk,\bot,\certout)$
    \end{receive}
    \begin{receive}[\dfesr]{keysharequery}{x,\PartyA,\PartyB}
      \Comment{A is corrupted}
      \State \Send*[$\env$]{$\msg{KeyShareGen}{x,\PartyA,\PartyB}$}{$\PartyA$} and \Await for $\msg{KeyShareGen}{x,\PartyA,\PartyB}$ from $\PartyA$ to $\SC^{\PartyB}$
      \State \Return $\msg{OK}{}$
    \end{receive}

  \begin{receive}[\dfesr]{keysharegen}{f,\PartyA,\PartyB}
    \If{$\PartyA$ is honest}
      \If{$\PartyB$ is honest} {}$\funcF\gets\leakfunc$\Else{ }$\funcF\gets f$\EndIf

        \State $(\sigpk,\sigsk,\certout)\gets\dfesrsimKS[\PartyA]$
        \State $\irsig{}\gets\sigscheme.\sigs(\sigsk,\funcF,\PartyB)$
        \State $\dfesrsimKShares[\funcF,\PartyA,\PartyB]\gets\irsig{}$
        \State \Send{KeyShareGen}{\funcF,\irsig{},\sigpk,\certout}{\SC^{\PartyB}_{\PartyA}}


    \EndIf
  \end{receive}
  \begin{receive}[\text{corrupted party } \party \text { to } \Gatt]{resume}{\eid{},\inp}
    \If{$\gattsent[\eid{}].\install[1]=\prog{\DE^{\certvk}} \land \inp[0]=\provisionkw $} 
    \State $(\provisionkw,\dsteelkshares_{\funcF},\eid{\funcF},\pkf,\irsig{\funcF},\funcF,\pid)\gets\inp$
    \For{$(\irsig{},\sigpk,\certout)\in\dsteelkshares_{\funcF}$}
    \State \Assert
    $(\sigpk,\cdot,\certout)=\dfesrsimKS[\funcF,\PartyA,\party]$ for
    some $\PartyA$
    \If{$\irsig{}\not\in\dfesrsimKShares[\funcF,\PartyA,\party]$}
      \State $\dfesrsimKShares[\funcF,\PartyA,\party]  \gets \irsig{}$
        \State
        \AsyncSendAndReceive*[$\PartyA$]{$\msg{KeyShareGen}{\funcF,\irsig{},\sigpk,\certout}$}{$\dfesr$}{$\msg{KeyShareQuery}{\funcF,\PartyA,\party}$} \Send{OK}{}{\dfesr}
    \EndIf
    \EndFor
    \EndIf
    \State \Send{resume}{\eid{},inp}{\fesrsim}
  \end{receive}
  \begin{receive}[\fesrsim\text{ on behalf of } \party]{encrypt}{\inp}
    \State
    \SendAndReceive*[$\party$]{$\msg{encrypt}{\inp}$}{$\dfesr$}{$\msg{encrypted}{\outp}$}
    \State \Send*[$\dfesr$]{$\msg{encrypted}{\outp}$}{$\fesrsim$}
  \end{receive}
  \begin{receive}[\fesrsim\text{ on behalf of } \party]{decrypt}{\inp}
    \State\SendAndReceive*[$\party$]{$\msg{decrypt}{\inp}$}{$\dfesr$}{$\msg{decrypted}{\outp}$}
    \State \Send*[$\dfesr$]{$\msg{decrypted}{\outp}$}{$\fesrsim$}
  \end{receive}
  \begin{receive}[none]{*}{}
    \State forward * to $\fesrsim$
  \end{receive}
\end{simulator}
We now show, via a series of hybrid experiments, that given the above simulator, the real and ideal worlds are
indistinguishable from the environment's viewpoint. We begin with the real-world protocol,
which can be considered as \emph{Hybrid} 0.

\emph{Hybrid} 1 consists of the ideal protocol for $\dfesr$, which includes
the relevant dummy parties, and the simulator $\dfesrsim'$, which on any message
from the environment ignores the output of the ideal functionality, and
faithfully reproduces protocol $\dsteel$. The equivalence between
\emph{Hybrid}s 0 and 1 is trivial due to the behaviour of $\dfesrsim'$.

\emph{Hybrid} 2 replaces all operations of $\dfesrsim'$ where the protocol
$\dsteel$ behaves in the same way as $\Steel$ (except that it sends messages
with the full set of arguments expected by $\dsteel$ rather than those in
$\Steel$, and receives the equivalent $\dsteel$ return values) with a call to an
emulated $\fesrsim$ as defined in Supporting Material ~\ref{app:steelsim}. Due
to the security proof of the $\Steel$ protocol in \cite{PKC:BKMT21}, we now use the
simulator of \emph{Hybrid} 2 to simulate $\fesr$ with respect to protocol $\Steel$,
making the two hybrids indistinguishable. An environment that is able to distinguish
between the two hybrids could create an adversary that can distinguish between
executions of $\fesr$ and $\Steel$; but due to the UC emulation statement, no
such environment can exist in the presence of $\fesrsim$. The reduction to
$\fesrsim$ greatly simplifies the current proof, as we are guaranteed the
security of the secure key provisioning and decryption due to the similarities
of these two phases of the protocols between $\Steel$ and $\dsteel$.

\emph{Hybrid} 3 modifies the simulator of \emph{Hybrid} 2 by replacing all the signature verification operations for
attestation signatures in $\dfesr$ with a table lookup from
$\fesrsim.\gattsigned$. The new table lookups for attestation signatures
complement the ones enacted by $\fesrsim$, while capturing behaviour that is
unique to $\dsteel$. Similar to \cite[Lemma~2]{PKC:BKMT21}, if the
environment can distinguish between this hybrid and the previous one, it can
construct an adversary to break the unforgeability of signatures.


\emph{Hybrid} 4 modifies the simulator of \emph{Hybrid} 3 by replacing \textsc{KeyShareGen} and \textsc{KeyShareQuery} requests for any
functions with a request for a dummy function (such as the natural leakage
function $\leakfunc$), for all these requests where both the
encryptor and the decryptor are honest. The environment is not able to
distinguish between the two hybrids due to the security of the secure channel
functionality (as defined in Supporting Material \ref{app::func_sc}): the secure
channel only leaks the length of a message exchanged between sender and
receiver, and assuming that we represent functions with a fixed-length string
(such as a hash of it's code), the leakage between this hybrid and the previous
one is indistinguishable.

\emph{Hybrid} 5 adds an additional check to the simulator of \emph{Hybrid} 4 before it can run the
$\provisionkw$ command on enclave $\prog{DE^{\certvk}}$ through the internal
$\fesrsim$ simulator. The check ensures that all keyshares passed by the
malicious decryptor to the enclave are signed by a party who has first
registered their verification key with the certificate authority. Then, if the
signature has not been generated through a call to the ideal functionality but
rather through a local signing operation, the simulator notifies the ideal functionality to
update its internal keyshare count. This hybrid essentially replaces the algorithmic
signature verification operations in the previous one with two table lookups
(for both verification key certification and keyshare authenticity). If an
adversary was able to bypass the checks by providing either a certificate that
wasn't produced by the ideal functionality or a keyshare that didn't match with
the triple of ($\funcF,\PartyA,\PartyB$), they would be able to create an
adversary that could break the unforgeability of signature scheme $\sigscheme$
in the same manner as in Hybrid~3.
Thus the hybrid is indistinguishable from the previous one.

The simulator defined in \emph{Hybrid} 5 is identical to $\dfesrsim$; thus, it
holds that $\dsteel$ UC-emulates $\dfesr$. 
  \hfill\(\Box\)
\end{proof}

\subsection{Turning Stateful functions into Multi-Input}\label{sec::mi-fe}
As pointed out in \cite{PKC:BKMT21}, $\fesr$ subsumes Multi-Input Functional
Encryption \cite{EC:GGGJKL14} in that it is possible to use the state to emulate
functions over multiple inputs. We now briefly outline how to realise a
Multi-Input functionality using $\fesr$ (and by extension $\dfesr$) through the
definition of a simple compiler from single-input stateful functions to
multi-input functions.
To compute a stateless, multi-input function
$\funcF : (\underbrace{\fedomain \times \dots\times \fedomain}_{n} ) \rightarrow \ferange$
we define the following stateful functionality:
\begin{algorithmic}
\Function{$\aggregfunc_{\funcF,n}$}{$\fein,\festate$}
\If{$|\festate| < n$} \Return ($\bot,\festate \concat \fein$)
\Else{} \Return $(\funcF(\festate || \fein),\initstate)$\algorithmiccomment{$\festate$ is equivalent to the array containing $\fein_{1},\dots,\fein_{n-1}$}
\EndIf
\EndFunction
\end{algorithmic}
where input $x$ is in $\fedomain$, the state $\festate$ is in $\festates$, and $n$ is bounded by the maximum size of $\festates$.
The above aggregator function is able to merge the inputs of multiple encryptors
because in $\fesr$ the state of a function is distinct between each decryptor;
therefore, multiple decryptors attempting to aggregate inputs will not interfere
with each other's functions.
%
%
There are several possible extensions to the above compiler:
\begin{enumerate}
  \item In $\aggregfunc_{\funcF,n}(\cdot)$, the order of parameters relies on the
decryptor's sequence of invocations. If $\funcF$ is a function where the order
of inputs affects the result,
malicious decryptor $\PartyB$ could choose not to run decryption in the same
order of inputs as received. It is possible to further extend the
decryption function to respect the order of parameters set by each encryptor. If
a subset of encryptors is malicious, we can parametrise the function by a set of
public keys for each party, and ask them to sign their inputs.
  \item\label{compiler:stateful} The compiler can be easily extended to multi-input \emph{stateful} functionalities, by
keeping a list of inputs (as a field) within the state array and not
discarding the state on the $n$-th invocation of the compiler.
  \item\label{compiler:variable} One additional advantage of implementing
        Multiple-Input functionalities through stateful functionalities is that
        we are not constrained to functions with a fixed number of inputs. If we
        treat the inner functionality to the compiler as a function taking as
        input a list, we can use the same compiler functionality for inner
        functions of any $n$-arity (we denote this type of functions as
        $\enmuldomain \rightarrow \enmulrange$). On the first (integer) input to
        the aggregator, we set it as a special field $n$ in the state, and for
        the next $n-1$ calls we simply append the inputs to the state, while
        returning the empty value. On the $n$th call, we execute the function on
        the stored state field (now containing $n$ entries), erase the state and
        index from memory and wait for the next call to set a new value to $n$.
\end{enumerate}

In the next section, we will use the compiler $\aggregsfunc_{\funcF}$, which
combines the above defined compiler extensions \ref{compiler:stateful} and
\ref{compiler:variable} to compute any stateful function
$\funcF$ with variable number of inputs
($\funcF : \enanaldomain \rightarrow \enanalrange$) from $\dfesr$.





\section{The $\glassvault$ platform}\label{sec::gv}

In this section, first, we provide the formal
definition of ``Analysis-augmented Exposure Notification'' ($\enplus$) which is
an extension of the standard Exposure Notification ($\enscheme$), proposed in
\cite{EPRINT:CKLRSSTVW20}, to allow arbitrary computation on data shared by
users. Then, we present $\glassvault$ and show that it
UC-realises the ideal functionality of $\enplus$, i.e., $\enplusideal$.

\vspace{-3mm}
\subsection{Analysis-augmented Exposure Notification ($\enplus$)}\label{sec::enplus}


Since $\enplus$ is built upon $\enscheme$, we first re-state relevant notions
used in the UC modelling of ${\enscheme}$. 
Specifically, $\enscheme$ relies on the time functionality $\timeidealfunc$ and the ``physical reality''
functionality $\phyidealfunc$ that we present in  Supporting Material \ref{app::func_time}
and \ref{app::func_phys}, respectively. In particular, $\phyidealfunc$ models
the occurrence of events in the physical world (e.g., users' motion or location
data). Measurements of a real-world event are sent as input from the
environment, and each party can retrieve a list of their own measurements. The
functionality only accepts new events if they are ``physically sensible'', and
can send the entire list of events to some privileged entities such as ideal
functionalities. Using $\timeidealfunc$ and $\phyidealfunc$ as subroutines, the
functionality of $\enscheme$, i.e., $\funct{\enscheme}$ (presented formally in
Supporting Material \ref{app::func_en}), is defined in terms of a risk estimation function
$\enriskfunc$, a leakage function $\enleakfunc$, a set of allowable measurement
error functions $\enerrorfuncset$, and a set of allowable faking functions
$\enfakingfuncset$. The functionality queries $\phyidealfunc$ and applies the
measurement error function chosen by the simulator to compute a ``noisy record
of reality''. This in turn is used to decide whether to mark a user as infected,
and to compute a risk estimation score for any user. The adversary can mark some
parties as corrupted and obtain leakage of their local state, as well as
modifying the physical reality record with a reality-faking function. This
allows simulation of adversarial behaviour such as relay attacks (where an
infectious user appears to be within transmission distance from a
non-infectious malicious user). 
%
The functionality captures a variety of contact tracing
protocols and attacker models via its parameters. For simplicity, it does
not model the testing process the users engage in to find out they are positive,
and it assumes that once a user is notified of exposure they are removed from
the protocol.

The extension to $\enplus$ involves an additional entity to the
above scheme, namely, analyst $\enanalyst$, who wants to learn a certain
function, $\enanalfunc$, on data contributed by exposed users, some of which
might be sensitive. Thus, exposed users are provided with a mechanism to accept
whether an analyst is allowed to receive the result of the executions of any
particular function. In order to receive the result, the analyst needs to be
authorised by a portion of exposed users determined by function
$\enkeythreshold$. We denote by $\phyrecsec$ a field in the physical reality
record for a user to be used for storing any sensitive data.

We now present a formal definition of the ideal functionality for $\enplus$.
Highlighted sections of the functionality represent where $\enplus$ diverges
from $\enscheme$.


\vspace{1mm}

\begin{functionality}{$\enplusideal[\enriskfunc,\enerrorfuncset,\enfakingfuncset,\enleakfunc,\enanalfuncset,\enkeythreshold]$}


  \noindent The functionality is parametrised by exposure risk function
  $\enriskfunc$, a set of allowable error functions $\enerrorfuncset$ for the
  physical reality record, a set of faking functions $\enfakingfuncset$ for the
  adversary to misrepresent the physical reality, and a leakage function
  $\enleakfunc$, as in the regular $\funct{\enscheme}$. ${\enanalfuncset}$ is
  the set of all functions
  $\{\enanalfunc\mid \enanalfunc:\enanaldomain\rightarrow \enanalrange\}$
  an analyst could be authorised to compute. $\enkeythreshold()$ is a function
  of the current number of users required to determine the minimum threshold of
  analyst authorisations.
  \begin{statedecl}
    \statevar{\ensharexprec}{}{List of users who have shared their exposure
      status and time of upload}
    \statevar{\enusers}{}{List of active users; $\ensharexprec \cap \enusers = \emptyset$ }
    \statevar{\enuserscorr}{}{List of corrupted users}
    \statevar{\enanalysts}{}{For each pair of analyst and allowed function, the dictionary $\enanalysts$ contains the users that have authorised this pair}
    \statevar{\enanalcorr}{}{Static set of corrupted analysts}
    \statevar{\enanalstate}{}{State table for function, analyst pairs}
  \end{statedecl}
  \begin{receive}[\adversary]{Setup}{\epsilon^{*}}
    \State \Assert $\epsilon^{*} \in \enerrorfuncset$
     \State $\ennoisyrec \gets \emptyset$\algorithmiccomment{Initialise noisy record of
      physical reality}
  \end{receive}
  \begin{receive}{ActivateMobileUser}{\enuser}
    \State $\enusers \gets \enusers \concat \enuser$
    \State \Send{ActivateMobileUser}{\enuser}{\adversary}
  \end{receive}
  \begin{receive}{ShareExposure}{\enuser}
    \State \SendAndReceive{AllMeas}{\epsilon^{*}}{\phyidealfunc}{\ennewnoisyrec}
    \State $\ennoisyrec \gets \ennoisyrec \concat \ennewnoisyrec$ 
    \If{$\ennoisyrec[\enuser][\phyrecinf] = \bot$}{} \Return $\enerrmsg$
    \Else

    \State \SendAndReceive{time}{}{\timeidealfunc}{\mathit{t}} \State
    $\ensharexprec \gets \ensharexprec \concat (\enuser,\mathit{t}); \enusers \gets \enusers \setminus
    \{\enuser\}$ 
    \If{\hl{$\enuser \in \enuserscorr$}}
    \Send{\hl{ShareExposure}}{$\hl{$\enuser,\ennoisyrec[\enuser][\phyrecsec]$}$}{\adversary}
    \Else{ }\Send{ShareExposure}{\enuser,\bot}{\adversary} \EndIf \EndIf
  \end{receive}
  \begin{receive}{ExposureCheck}{\enuser}
    \If{$\enuser \in \enusers$}
    \State \SendAndReceive{AllMeas}{\epsilon^{*}}{\phyidealfunc}{\ennewnoisyrec}
    \State $\ennoisyrec \gets \ennoisyrec \concat \ennewnoisyrec ;\mu \gets \ennoisyrec[\enuser] \concat \ennoisyrec[\ensharexprec]$
      \State \Return $\enriskfunc(\enuser,\mu)$
    \Else{} \Return $\enerrmsg$
    \EndIf
  \end{receive}
  \begin{receive}[\enanalyst]{\hl{RegisterAnalyst}}{ $\hl{\ $ \enanalfunc,\enanalyst$}$}
    \If{\hl{$\enanalfunc \in \enanalfuncset$}}

      \State \Send{\hl{RegisterAnalyst}}{$\hl{\ $\enanalfunc,\enanalyst$}$}{\adversary}
      \ForAll {\hl{$\enuser\in\ensharexprec$}} \Send{\hl{RegisterAnalystRequest}}{$\hl{\ $\enanalfunc,\enanalyst,\enuser$}$}{\enuser}\EndFor
      \State \hl{$\enanalysts[\enanalfunc,\enanalyst] \gets []$}
    \EndIf
  \end{receive}
  \begin{receive}[\enuser]{\hl{RegisterAnalystAccept}}{$\hl{\ $\enanalfunc,\enanalyst,\enuser$}$}
    \If{\hl{$\enuser\in\enuserscorr\lor\enanalyst\in\enanalcorr$}} \AsyncSendAndReceive{\hl{RegisterAnalystAccept}}{$\hl{$\enanalfunc,\enanalyst,\enuser$}$}{\adversary}{\msg{OK}{}}
    \EndIf
    \State
    \hl{$\enanalysts[\enanalfunc,\enanalyst] \gets \enanalysts[\enanalfunc,\enanalyst]\concat\enuser$}
    \State \Send{\hl{RegisterAnalystAccept}}{$\hl{\ $\enuser,\enanalfunc$}$}{\enanalyst}
  \end{receive}
%
  \begin{receive}{\hl{Analyse}}{$\hl{\ $\enanalfunc$}$}
    \If{\hl{\mbox{$|\enanalysts[\party,\enanalfunc]| \geq \enkeythreshold(|\ensharexprec|)$}}}
      \State
      \hl{\mbox{$(\feout,\enanalstate') \gets \enanalfunc(\ennoisyrec[\ensharexprec][\phyrecsec],\enanalstate[\enanalfunc,\enanalyst])$}}
      \State \hl{$\enanalstate[\enanalfunc,\enanalyst]\gets\enanalstate'$}
          \If{\hl{$\party\in\enanalcorr$}} \AsyncSendAndReceive{\hl{Analysed}}{$\hl{$\enanalfunc,\party,y$}$}{\adversary}{OK}\EndIf

      \State \Return \hl{$\feout$}
    \EndIf
  \end{receive}
  \begin{receive}{RemoveMobileUser}{\enuser}
      \State $\enusers \gets \enusers \setminus \{\enuser\}$
  \end{receive}
  %
  \begin{receive}[\adversary]{Corrupt}{\enuser}
    \State $\enuserscorr \gets \enuserscorr \concat \enuser$
    \State \Return \hl{$\{(\enanalfunc,\enanalyst ): \enuser\in\enanalysts[\enanalfunc,\enanalyst ]\}$}
  \end{receive}
  \begin{receive}[\adversary]{MyCurrentMeas}{\enuser,\phyrecfields,\phyerrfunc}
    \If{$\enuser \in \enuserscorr$}
    \State \SendAndReceive{MyCurrentMeas}{\enuser,\phyrecfields,\phyerrfunc}{\phyidealfunc}{\phymeasret}
    \State \Send{MyCurrentMeas}{\phymeasret}{\adversary}
    \EndIf
  \end{receive}
  \begin{receive}[\adversary]{FakeReality}{\phi}
    \If{$\phi \in \enfakingfuncset$}{} $\ennoisyrec \gets \phi(\ennoisyrec)$
    \EndIf
  \end{receive}
  \begin{receive}[\adversary]{Leak}{}
    \State \Send{Leak}{\enleakfunc(\{\ennoisyrec, \enusers, \ensharexprec\})}{\adversary}
  \end{receive}
  \begin{receive}[\environment]{IsCorrupt}{\enuser}
    \State \Return $\enuser \stackrel{?}{\in} \enuserscorr$
  \end{receive}
\end{functionality}

\subsection{$\glassvault$ protocol}\label{subsec::glassvault}

In this section, we present the $\glassvault$ protocol. It is a delicate
combination of two primary primitives; namely, (i) the original exposure
notification proposed by \citet{EPRINT:CKLRSSTVW20}, and (ii) the enhanced
functional encryption that we proposed in Section \ref{sec::DD-FESR}. In this
protocol, infected users upload not only the regular data needed for exposure
notification but also the encryption of their sensitive measurements (e.g.,
their GPS coordinates, electronic health records, environment's air quality).
Once an analyst requests to execute some specific computations on exposed users'
data, the users are informed via public announcements. At this stage, users can
provide permission tokens to the analyst (in the form of functional keysahres),
who can run such computations only if the number of tokens exceeds a threshold
defined by the function $\enkeythreshold$ over the number of parties in the set
$\PartyAS$ of $\dfesr$ (note that since a party is registered to $\dfesr$ only
if they are exposed, the number of exposed users matches the size of
$\PartyAS$). Due to the security of the proposed functional encryption scheme,
the $\glassvault$ analyst does not learn anything about the users' sensitive
inputs beyond what the function evaluation reveals.

It is not hard to see that this protocol (a) is generic, as it supports
arbitrary secure computations (i.e., multi-input stateful and randomised
functions) on users' shared data, and (b) is transparent, as computations
are performed only if permission is granted by a sufficient number of users and
in that the user can choose whether they are willing to share their sensitive data
or not, making the collection of information consensual.
Besides the $\enscheme$ and $\dfesr$ functionality, the $\glassvault$ protocol makes use of the
physical reality functionality $\phyidealfunc$, the exposure notification
functionality $\enidealfunc$, and the trusted bulletin board functionality
$\tbbidealfunc$, described in Supporting Material \ref{app::func_phys}, \ref{app::func_en},
and~\ref{app::func_tbb}, respectively.

\vspace{1mm}


%
%


%

\begin{protocol}{$\glassvault[\enriskfunc,\enerrorfuncset,\enfakingfuncset,\enleakfunc,\enanalfuncset,\enkeythreshold,\PartyC]$} 


  \noindent The protocol takes the same class of parameters as defined in
  $\enplusideal$, and the identity of a trusted authority $\PartyC$. We use
  $\enuser$ to refer to a normal user of the Exposure Notification System
  (corresponding to $\PartyA$ in $\dfesr$). We use $\enanalyst$ to refer to an
  analyst (corresponding to $\dfesr$'s decryptor, $\PartyB$). Among other ideal
  setups, $\glassvault$ leverages the exposure notification ideal functionality
  $\enscheme[\enriskfunc,\enerrorfuncset,\enfakingfuncset,\enleakfunc]$ and
  functional encryption ideal functionality
  $\dfesr[\enanalfuncset,\emptyset,\emptyset,\PartyC]$.

    \noindent\underline{User $\enuser$:}
    \begin{receive}{ActivateMobileUser}{}

      \State \Send{ActivateMobileUser}{\enuser}{\enidealfunc}
    \end{receive}
    \begin{receive}{ShareExposure}{}
    \State \SendAndReceive{ShareExposure}{\enuser}{\enidealfunc}{r}
    \If{$r \neq \enerrmsg$}
    \State \Send{Setup}{}{\dfesr}
    \State \SendAndReceive{MyCurrentMeas}{\enuser,\phyrecsec,\phyerrfunc}{\phyidealfunc}{\phyusermeas^{\phyerrfunc}_{\phyrecsec}}
    \State \SendAndReceive{Encrypt}{\phyusermeas^{\phyerrfunc}_{\phyrecsec},\enkeythreshold(|\dfesr.\PartyAS|)}{\dfesr}{\msg{encrypted}{\handle}}
    \State \textbf{erase} $\phyusermeas^{\phyerrfunc}_{\phyrecsec}$ \textbf{and} \Send{Add}{\handle}{\tbbidealfunc}
    \EndIf
    \end{receive}
    \begin{receive}{ExposureCheck}{}
      \State \SendAndReceive{ExposureCheck}{\enuser}{\enidealfunc}{\rho_{\enuser}}
      \If{$\rho_{\enuser} \neq \enerrmsg$}{} \Return$\rho_{\enuser}$\EndIf
    \end{receive}

    \begin{receive}{RegisterAnalystAccept}{\enanalfunc,\enanalyst}
      \State \Send{KeyShareGen}{\aggregsfunc_{\enanalfunc},\enanalyst}{\dfesr}
      \State \Send{RegisterAnalystAccept}{\enuser,\enanalfunc}{\enanalyst}
    \end{receive}
    \underline{Analyst $\enanalyst$:}
    \begin{receive}{RegisterAnalyst}{\enanalfunc,\enanalyst}
    \State \Send{Setup}{}{\dfesr}
    \ForAll{exposed $\enuser$}{} \Send{RegisterAnalystRequest}{\enanalfunc,\enanalyst,\enuser}{\party}\EndFor
    \end{receive}
    \begin{receive}{Analyse}{\enanalfunc}
      \State \SendAndReceive{Retrieve}{}{\tbbidealfunc}{\tbbrec}
      \State \SendAndReceive{Encrypt}{(|\tbbrec|,0)}{\dfesr}{\handle_n}
      \State \SendAndReceive{Decrypt}{\handle_n,\aggregsfunc_{\enanalfunc}}{\dfesr}{\msg{decrypted}{|\tbbrec|}}
      \For{$\handle\in\tbbrec$}
      \State \SendAndReceive{Decrypt}{\handle,\aggregsfunc_{\enanalfunc}}{\dfesr}{\msg{decrypted}{\feout}}
      \EndFor
      \If{$\feout\neq\bot$}{} \Return$\msg{Decrypted}{\enanalfunc,\party,\feout}$\EndIf
    \end{receive}
\end{protocol}


\subsection{Proof of security}
We now show that $\glassvault$ is secure:
\begin{theorem}\label{thm::glassvault}
  Let
$\enriskfunc,\enerrorfuncset,\enfakingfuncset,\enfakingfuncset^{+},\enleakfunc,\enleakfunc^{+},\enanalfuncset,\enkeythreshold,$ and $\PartyC$
  be parameters such that the following conditions
  hold: 
  (1) $\enfakingfuncset \subset \enfakingfuncset^{+}$,
  (2) for every input $x$,
  it holds that\footnote{Note, if $x$
    is a labelled dictionary and $\enleakfunc$ returns a dictionary which
    includes a subset of entries in $x$ and optionally any other additional
    records, $\enleakfunc^{+}$ strictly returns more records than $\enleakfunc$.
    }
    $\enleakfunc(x)=\enleakfunc(\enleakfunc^{+}(x))$, and
 (3) there is a function
  $\phi^{+}\in\enfakingfuncset^{+}$ such that for every input $x$, every noisy
  physical reality record $\ennoisyrec\in x$, every function
  $\phi\in\enfakingfuncset$, and every set of users $\enuser$,
    (3.1) $\phi(\ennoisyrec)$ does not tamper with
      sensitive data record  $\ennoisyrec[\enuser][\phyrecsec]$ but $\phi^{+}(\ennoisyrec)$
      does, and
    (3.2)
      $\enleakfunc(x)$ does not contain any instruction to leak the contents of $\ennoisyrec[\enuser][\phyrecsec]$ but
      $\enleakfunc^{+}(x)$ does.
%
 Then, it holds that
\[\glassvault[\enriskfunc,\enerrorfuncset,\enfakingfuncset,\enleakfunc,\enanalfuncset,\enkeythreshold,\PartyC] \textit{ UC-realises } \enplusideal[\enriskfunc,\enerrorfuncset,\enfakingfuncset^{+},\enleakfunc^{+},\enanalfuncset,\enkeythreshold],  \]
in the presence of global functionalities $\timeidealfunc$ and $\phyidealfunc$.


%
\end{theorem}
\begin{proof}
  We construct a simulator $\simglassvault$ to prove Theorem \ref{thm::glassvault}. The high-level task of our simulator is to \emph{synchronise} the inputs of
  the analysis functions between the ideal world (where they are stored in
  $\phyidealfunc$), and the real world (where they are held in the $\dfesr$
  ideal
  repository). 
  The simulator updates a simulated trusted bulletin board by obtaining, through
  the leakage function, the secret data for all honest users who have shared
  their exposure, and encrypting it through $\dfesr$.

  When any registered and corrupted analyst executes an \textsc{Analyse}
  request in the ideal world, the simulator allows the ideal functionality to
  return the ideal result of the computation only if the adversary instructs the
  analyst to correctly aggregate the ciphertexts stored in the bulletin board
  through $\dfesr$ decryption requests to the appropriate aggregator function.
%
   We also simulate the \textsc{RegisterAnalyst} and
  \textsc{RegisterAnalystAccept} sequence of operations by triggering the
  corresponding \textsc{Setup} and \textsc{KeyShareGen} subroutines in $\dfesr$.
  Any other adversarial calls to $\enplusideal$ such as
  $\msg{Setup}{\epsilon^{*}}$ and $\msg{FakeReality}{\phi}$ are allowed and
  redirected to $\enidealfunc$, as long as $\epsilon^{*}\in\enerrorfuncset$ and
  $\phi\in\enfakingfuncset)$.

\begin{simulator}{$\simglassvault$}
  \begin{statedecl}

   \statevar{\enleakfunc^+}{}{function to leak anything that $\enleakfunc$ does,
     as well as the contents of $\phyrecsec$ for all users}
   \statevar{\ensimtbb}{\{\}}{Table that stores messages uploaded to the Trusted
   Bulletin Board}
    \statevar{\enuserscorr}{}{List of corrupted users}
    \statevar{\ensharexprec}{}{List of exposed users}
  \end{statedecl}
  \begin{receive}[\enplusideal]{ShareExposure}{\enuser,\phyusermeas^{\phyerrfunc}_{\phyrecsec}}
    \State \Send*[$\enuser$][sim]{\msg{Setup}{}}{$\dfesr$}
    \If{$\phyusermeas^{\phyerrfunc}_{\phyrecsec}=\bot$}
    \Comment{simulate honest user:}
      \State \SendAndReceive{Leak}{}{\enplusideal}{\mathsf{r}}
      \State $\phyusermeas^{\phyerrfunc}_{\phyrecsec} \gets \mathsf{r}[\enuser][\phyrecsec]$

    \EndIf

    \State \SendAndReceive*[$\enuser$][sim]{$\msg{Encrypt}{\phyusermeas^{\phyerrfunc}_{\phyrecsec},\enkeythreshold(|\dfesr.\PartyAS|)}$}{$\dfesr$}{$\msg{encrypted}{\handle}$}
    \State $\ensharexprec\gets\ensharexprec\concat\enuser$; $\ensimtbb \gets \ensimtbb\concat{\handle}$
  \end{receive}
  \begin{receive}[\enplusideal]{RegisterAnalyst}{\enanalfunc,\enanalyst}
    \State \Send*[$\enanalyst$][sim]{\textsc{setup}}{$\dfesr$}
    \For{$\enuser\in\ensharexprec$} \Send*[$\enanalyst$][sim]{$\msg{RegisterAnalystRequest}{\enanalfunc,\enanalyst,\enuser}$}{$\env$}\EndFor
  \end{receive}
  \begin{receive}[\enplusideal]{RegisterAnalystAccept}{\enanalfunc,\enanalyst,\enuser}
    \If{$\enuser\in\enuserscorr$}
      \State \Await for
      $\msg{KeyShareGen}{\aggregsfunc_{\enanalfunc},\enanalyst}$ from $\enuser$
      to $\dfesr$
    \Else
    \State \Send*[$\enuser$][sim]{$\msg{KeyShareGen}{\aggregsfunc_{\enanalfunc},\enanalyst}$}{$\dfesr$}
    \EndIf
    \State \Return OK
  \end{receive}
  \begin{receive}[\enplusideal]{Analysed}{\enanalfunc,\party,\feout}

      \State \Await for $\msg{Encrypt}{(|\ensimtbb|,0)}$ from $\party$ to
      $\dfesr$ and for response $\msg{encrypted}{\handle_{n}}$

      \State \Await for $\msg{Decrypt}{\handle_{n},\aggregsfunc_{ \enanalfunc}}$
      from $\party$ to
      $\dfesr$ and for response $\msg{decrypted}{|\ensimtbb|}$
      \For{$\handle\in\ensimtbb$}
      \State \Await for $\msg{Decrypt}{\handle,\aggregsfunc_{ \enanalfunc}}$
      from $\party$ to
      $\dfesr$ and for response $\msg{decrypted}{\_}$
      \EndFor
      \State \Return OK
  \end{receive}

  \begin{receive}[\adversary \text{ to }\enidealfunc]{Leak}{}
    \State \SendAndReceive{Leak}{}{\enplusideal}{\mathsf{r}}
    \State \Return $\enleakfunc(\mathsf{r})$
  \end{receive}
  \begin{receive}[\adversary \text{ to }\enidealfunc]{FakeReality}{\phi}
    \State \Assert $\phi\in\enfakingfuncset$
    \State \Send{FakeReality}{\phi}{\enplusideal}
  \end{receive}
  \begin{receive}[\adversary \text{ to }\enidealfunc]{Corrupt}{\enuser}
    \State $\enuserscorr \gets \enuserscorr \concat \enuser$
    \State \SendAndReceive{Corrupt}{\enuser}{\enplusideal}{\enanalysts_\enuser}
    \For{$(\enanalfunc,\enanalyst)\in\enanalysts_{\enuser}$}
    \State \Send*[$\enuser$][sim]{$\msg{KeyShareGen}{\aggregsfunc_{\enanalfunc},\enanalyst}$}{$\dfesr$}
    \EndFor
  \end{receive}
  \begin{receive}[\adversary \text{ to }\enidealfunc]{*}{}
   \State \Send{*}{}{\enplusideal}
  \end{receive}
  \begin{receive}[\adversary\text{ to }\tbbidealfunc]{Retrieve}{}
    \State \Return $\ensimtbb$
  \end{receive}
\end{simulator}

  We prove that the $\glassvault$ protocol is secure under a model of semi-honest
  adversarial behaviour; this means showing that $\mathsf{view}^{\env}_{\mathsf{REAL}}\approx\mathsf{view}^{\env}_{\mathsf{IDEAL}}$.

  We argue that for all messages sent by the environment, the ideal world
  simulator produces a view that is indistinguishable from $\mathsf{view}^{\env}_{\mathsf{REAL}}$. 
 In particular:
  \begin{itemize}
    \item \textsc{ShareExposure}: when a user $\enuser$ shares their exposure
          status, $\simglassvault$ is activated. If $\enuser$ is corrupted, it
          additionally receives the noisy record of $\enuser$'s sensitive data,
          $\phyusermeas^{\phyerrfunc}_{\phyrecsec}$.
         For honest users, $\simglassvault$ obtains the same sensitive records
         by using the $\textsc{Leak}$ function. 
         As in the real world, $\dfesr$ is invoked to encrypt the sensitive
          data, and the resulting handle is stored in an emulated trusted
          bulletin board. In both worlds, the array of stored handles follows a
          similar distribution as they are both generated by $\dfesr$ for the
          same messages. All other behaviour of \textsc{ShareExposure} is
          handled by $\enplusideal$ in the same way that $\enidealfunc$ would.

    \item \textsc{RegisterAnalyst}: when analyst $\enanalyst$ requests
          permission to compute a function $\enanalfunc\in\enanalfuncset$,
          the simulator registers them as a decryptor in $\dfesr$ (the same
          analyst can request to be registered multiple times, but the $\dfesr$
          functionality will ignore all but the first request). $\simglassvault$
          then emulates a request for \textsc{RegisterAnalystRequest} for all
          exposed users. 
          For the semi-honest case, when both honest and corrupted users are
          allowed by the environment to accept the request for a function, they
          will ask $\dfesr$ for a \textsc{KeyShareGen}. $\simglassvault$ learns
          abouts these \textsc{RegisterAnalystAccept} calls when either the analyst or user is
          corrupted. In the former case $\simglassvault$
          proactively sends the request to $\dfesr$ on behalf of the user, while
          in the latter it waits for the adversary to trigger the request. 
          If both user and analyst were honest, the adversary (in either world)
          should not learn that a request was granted. However, given we are in
          the dynamic user corruption setting, the simulator has to handle
          keyshare generation for newly corrupted users, by requesting keyshare
          generation to $\dfesr$ for all functions they had authorised pre-corruption.
    \item \textsc{Analyse}: for a corrupted analyst, $\simglassvault$ will
    ensure that they sync
    the state of the ideal function $\enanalfunc$ with that of the aggregated
    $\aggregsfunc_{\enanalfunc}$ in $\dfesr$. To aggregate all inputs stored in
    its emulated trusted bulletin board $\ensimtbb$, the analyst first encrypts the
    integer equal to the size of $\ensimtbb$ and passes it for decryption to
    $\aggregsfunc_{\enanalfunc}$ to initialise it. Then, for all handles stored
    in $\ensimtbb$, it also decrypts the corresponding values to the same aggregator. When all
    the decryptions have occurred, the final returned value will be the
    evaluation of $\enanalfunc$ on the sensitive data of all exposed users;
    $\simglassvault$ ignores this value and yields back to $\enplusideal$, which
    will return the ideal world result to the analyst.

    Since the inputs to the aggregator (the set of uploaded sensitive data
    to the trusted bulletin board) in the real world fully correspond to the
    input of $\enanalfunc$ in the ideal, the distributions of states and outputs
    for $\aggregsfunc_{\enanalfunc}$ in $\dfesr$
    and for $\enanalfunc$ in $\enplusideal$ are indistinguishable.

    \item \textsc{Leak}: on an adversarial request to learn some values from the combination
    of noisy record of reality, exposed users, and corrupted users,
    $\simglassvault$ obtains the corresponding leakage $\mathsf{r}$ from $\enplusideal$, and
    filters it by the admissible leakage $\enleakfunc$ for $\enidealfunc$.
    \item \textsc{FakeReality}: $\simglassvault$ ensures that the request to
    modify the noisy record of reality through $\enplusideal$ is also admissible
    in the real world with $\enidealfunc$.
  \end{itemize}
  All other messages are handled by redirecting them from
  $\enplusideal$ to $\enidealfunc$, since both functionalities behave in the
  same manner outside the cases we have already outlined.
%
%
  \hfill\(\Box\)
\end{proof}

In the rest of this section, we make a few additional remarks. At a high level, the Glass-Vault protocol can support any computation in a
privacy-preserving manner, in the sense that nothing beyond the computation
result is revealed to the analyst; more formally, in the simulation-based model,
a corrupted party’s view of the protocol execution can be simulated given only
its input and output. The Glass-Vault protocol can be considered as an
\emph{interpreter} that takes a description of any multi-input functionality
along with a set of inputs, executes the functionality on the inputs and returns
only the result to a potentially semi-honest analyst.

  While the above simulator guarantees security in the semi-honest setting, it
  is possible to design a simulator that allows corrupted parties to diverge
  from the protocol. In particular, this simulator would need to handle the case
  of a malicious user who encrypts via $\dfesr$ dishonestly generated data (in
  that it does not match with the corrupted user's physical reality
  measurements). Following the lead of \citet{EPRINT:CKLRSSTVW20}, we can
  account for these malicious ciphertexts by using functions in
  $\enfakingfuncset^{+}$ to modify the noisy record of physical reality in
  $\enplusideal$. Note that this simulation strategy imposes additional deviations
  from the physical reality beyond those unavoidably inherited by $\enfakingfuncset$
  due to its own simulation needs. This makes it harder to justify the usage of the
  protocol by an analyst who is interested in the correctness of the data
  processing; hence, our choice of adversarial model.


  Recall that the primary reason the Glass-Vault protocol offers
  ``transparency'' is that it lets users have a chance to decide which
  computation should be executed on their sensitive data. This is of particular
  importance because the result of any secure computation (including functional
  encryption) would reveal some information about the computation's inputs.
  However, having such an interesting feature introduces a trade-off: if many
  users value their privacy and decide not to share access to their data, the
  analyst may not get enough to produce any useful results, foregoing the
  collective benefits this kind of data sharing can engender
  \cite{schwartz1997privacy}. One way to solve such a dilemma would be to
  integrate a (e.g., blockchain-based) mechanism to incentivise users to grant
  access to their data for such computations; however, even then careful
  considerations are required, as the framing of why a user is asked to disclose
  their information can impact how much they value privacy \cite{acquisti2013privacy}.



\subsection{Cost evaluation}

In the $\glassvault$ protocol, an infected user's computation and communication
complexity for the proposed data analytics purposes is independent of the total
number of users, while it is linear with the number of functions requested by
the analysts.
An analyst's computation overhead depends on each function's complexity and the
number of decryptions (as each decryption updates the function's state). The
cost to non-infected users is comparable to the most efficient $\enscheme$ protocols that
realise $\enidealfunc$ (such as the protocol in~\cite[Section
8.5]{EPRINT:CKLRSSTVW20}): besides passively collecting measurements
of sensitive data, no other data-analytics operation is required until the
\textsc{ShareExposure} phase.

While the costliest component of the protocol is the functional encryption
module, it is possible to build an efficient implementation of $\glassvault$ due
to the construction of $\dsteel$, which relies on efficient operations
facilitated by trusted hardware.





\section{Example: infections heatmap}\label{sec::heatmap}
\newcommand{\heatmapfunc}{\mathsf{Heatmap}}
\newcommand{\heatmapsize}{k}
\newcommand{\heatmapmin}{q}
\newcommand{\heatmapkept}{n}
\newcommand{\heatmapparam}{\heatmapsize,\heatmapmin}
\newcommand{\heatmapinp}{u}
\newcommand{\heatmapoutp}{\vec{out}}
\newcommand{\heatmapstate}{\mathsf{state}}
\newcommand{\heatmapmap}{\mathsf{m}}
\newcommand{\heatmapbuf}{\mathsf{b}}
\newcommand{\heatmaptime}{\vec{t}}

In this section, we provide a concrete example of a computation that a
$\glassvault$ analyst
 can perform:
%
%
 generating a daily heatmap of the current clusters of infections. This is an
 interesting application of $\glassvault$, as it relies on collecting
 highly-sensitive location information from infected individuals. 
 

%
$\heatmapfunc_{\heatmapparam}(\fein,\festate)$ is defined as a multi-input stateful
function, parametrised by $\heatmapsize$: the number of distinct cells we divide
the map into, and $\heatmapmin$: the minimum number of exposed users that 
have shared their data.
%
 %
 The values of
these parameters affect the granularity of the results, computational costs,
and privacy of the exposed users.   Thus, they need to be approved as part of
the \textsc{KeyShareGen} procedure (in that the parameters' values are hardcoded
in the $\heatmapfunc$ program, so that different parameter values require
different functional keys).

Given the full set of exposed users' sensitive data, the $\heatmapfunc$ function
filters it to the exposed users' location history for the last $\incubt$ days
(the maximum number of days since they might have been spreading the virus due
to its incubation period), and constructs a list of
$\incubt \times \heatmapsize$ matrices, where an entry in each matrix
$\heatmapinp$ contains the number of hours within a day an infected individual
spent in a particular location. Location data is collected once every hour by
the user's phone, and divided into $\heatmapsize$ bins. The $\heatmapinp$ matrix
rows are in reverse chronological order, with the last row of the matrix
corresponding to the locations during the most recent day and each row above in
decreasing order until the first row which contains the locations $\incubt$ days
ago.


%
$\heatmapfunc$ maintains as part of its state a list, $\heatmapmap$, of
$\incubt$-sized circular buffers (a FIFO data structure of size $\incubt$; once
more than $\incubt$ entries have been filled, the buffer starts overwriting data
starting from the oldest entry). On every call with input $\fein$, the function
allocates a new circular buffer $\heatmapbuf$ for each matrix $\heatmapinp$ it
constructed from $\fein$, and assigns each of $\heatmapinp$'s rows to one of
$\heatmapbuf$'s $\incubt$ elements, starting from the top row. Each element in
$\heatmapbuf$ now contains a list of $\heatmapsize$ geolocations for a specific
day, with the first element containing the locations $\incubt$ days ago, and so
on. For any buffer already in $\heatmapmap$, we append a new zero vector,
effectively erasing the record of that user's location for the earliest day. If
there is a buffer that is completely zeroed out by this operation, we remove it
from $\heatmapmap$.

If we  have $|\heatmapmap|\geq\heatmapmin$, we return
the row-wise sum of vectors
$\sum\limits_{\scriptscriptstyle\heatmapbuf\in\heatmapmap}\sum\limits_{\scriptscriptstyle i=0}^{\scriptscriptstyle \incubt-1} \heatmapbuf[i]$.
The result is a single $\heatmapsize$-sized vector containing the total number of
hours spent by all users within the last $\incubt$ days: our heatmap. The full
pseudocode for the function is provided in Supporting material~\ref{app::heat}.
For the results' correctness, an analyst should run the function (through a
decryption operation) once a day. As the summation of user location vectors is a
destructive operation, the probability that a malicious analyst can recover any
specific user's input will be inversely proportional to $\heatmapmin$.

While the security proof of Theorem \ref{thm::glassvault} is in the semi-honest
setting, a fully malicious setting is not unrealistic, as
a user can tamper with their own client applications to upload malicious data
(terrorist attack \cite{EPRINT:Vaudenay20b}). Since there is no secure pipeline
from the raw measurements from sensors to a specific application, unless we
adopt the strong requirement that every client also runs a TEE (as
in~\cite{cryptoeprint:2020:1546}), it is impossible to certify that the users'
inputs are valid. Like most other remote computation systems, $\glassvault$
cannot provide blanket protection against this kind of attack. However, due
to its generality, $\glassvault$ allows analysts to use functions that include
``sanity checks'' to ensure that the data being uploaded are at least sensible,
in order to limit the damage that the attack may cause. In the heatmap case, one
such check could be verifying that for each row of $\heatmapinp$, it must hold
that its column-wise sum is equal to 24, since each row represents the number of
hours spent across various locations by the user in a day (assuming the user's
phone is on and able to collect their location at least once an hour throughout
a day). To capture this type of attack in the ideal functionality
$\enplusideal$, we instantiate it with a \textsc{FakeReality} function in
$\enfakingfuncset^{+}$ such that, if a malicious user $\enuser$ uploads this
type of fake geolocation, it will update $\enuser$'s position within the noisy
record of physical reality to match $\enuser$'s claimed location, while making
sure that other users who compute risk exposure and have been in close contact
with $\enuser$ will still be notified.

We highlight that \citet{EPRINT:BHKRW20} propose an ad-hoc scheme that produces
similar output. Their scheme relies on combining infection data provided by
health authorities with the mass collection of cell phone location data from
mobile phone operators. Unlike $\glassvault$, with its strong level of
transparency, the approach in~\cite{EPRINT:BHKRW20} does not support any
mechanism that allows the subjects of data collection to provide their direct
consent and opt-out of the computation.
\section{Future directions}\label{sec::future}



There are several possible directions for future research. An immediate goal
would be to implement $\glassvault$ for various data analytics, examine their
run-time, and optimise the system's bottlenecks. Since $\glassvault$'s analytics
results could influence public policy, it is interesting to investigate how this
platform could be equipped with mechanisms that allow result recipients to
verify the authenticity of outputs provided by the analyst. Another appealing
future research direction is to investigate the design of privacy-preserving
contact graph analysis as an application of $\glassvault$, which would let an
analyst construct a contact (sub)graph, using users' data, with minimum leakage
and maximum transparency, bridging the gap between centralised and
decentralised contact tracing.
%
%



\section*{Acknowledgments}

Lorenzo Martinico was supported by the National Cyber
Security Centre, the UK Research Institute in Secure Hardware and Embedded Systems
(RISE).  Aydin Abadi was supported in part by The National Research Centre on Privacy, Harm Reduction and Adversarial Influence Online (REPHRAIN), under UKRI grant: EP/V011189/1. 

\bibliography{abbrev2,crypto,extra}

\appendix

\section{Syntax of Functional Encryption}\label{app::fe}

  $\fe$ is defined over a class of functions
  $\Ffamily=\{\funcF\mid \funcF : \fedomain \rightarrow \ferange\}$, where $\fedomain$ is the domain and $\ferange$ is the range, consisting of the following
  algorithms:
\begin{itemize}
\item \underline{$\fes$} (run by $\PartyC\in\PartyCS$). It takes a security parameter $1^{\secp}$ as input and outputs a master keypair $(\mpk,\msk)$.
\item  \underline{$\fekg$} (run by $\PartyC\in\PartyCS$). It  takes $\msk$ and a function's description $\funcF \in \Ffamily$ as inputs
 and outputs functional key~$\funcKeyF$.
\item  \underline{$\fee$} (run by $\PartyA\in\PartyAS$). It takes a plaintext string $\fein \in \fedomain$ and $\mpk$ as inputs. It returns
 a ciphertext $\ct$ or an error.
\item  \underline{$\fed$} (run by $\PartyB\in\PartyBS$). It takes
ciphertext $\ct$ and functional key $\funcKeyF$ as inputs and returns a value $\feout\in\ferange$.
\end{itemize}

Informally, correctly evaluating the decryption operation on a ciphertext
$\ct \gets \fee(\mpk,\fein)$ using functional key $\funcKeyF$ should result in
$\feout \gets \funcF(\fein)$. The party $\PartyB$ should not learn anything
about $\PartyA$'s input, except for any information that $\feout$ reveals
about $\fein$ and some natural leakage from the ciphertext,
e.g., the length of the ciphertext.


\section{Background Functionalities}\label{app::func}
In this section, we provide an overview of the functionalities that are invoked
in the description of $\dsteel$, $\enplus$, and $\glassvault$ (presented in
Subsections~\ref{subsec::DD-Steel}, \ref{sec::enplus}, and
\ref{subsec::glassvault}, respectively). The overview is sufficient for the
clarification of the communication interface among said functionalities and the
interacting entities. For a detailed description of each background
functionality, we refer readers to the relevant work, except the exposure
notification functionality $\funct{\enscheme}$, which we present in detail as
pseudocode to allow an easy comparison with the extended functionality
$\enplusideal$.

\subsection{The global attestation functionality $\Gatt$}\label{app::func_Gatt}
The ideal functionality $\Gatt$ was introduced in~\cite{EC:PasShiTra17} and can be seen as an abstraction for a broad class of attested execution
processors. The functionality operates as follows:
\begin{itemize}
\item On message \textsc{initialize} from a party $\party$, it generates a pair of signing and verification keys $(\mathsf{spk},\mathsf{ssk})$. It stores $\mathsf{spk}$ as the master verification key $\gvk$, available to enclave programs, and $\mathsf{ssk}$ as the master secret key $\mathsf{msk}$, protected by the hardware. 
\item On message \textsc{getpk} from a party $\party$, it returns $\gvk$.
  \item On message $(\textsc{install},\idx,\mathsf{prog})$ from a (registered
        and honest) party $\party$, it asserts that $\idx$ corresponds to the
        calling party's session id. Then, it creates a unique enclave identifier
        $\mathsf{eid}$ and establishes a software enclave for
        $(\mathsf{eid},\party)$ as $(\mathsf{idx},\mathsf{prof},\emptyset)$. It
        provides $\party$ with $\mathsf{eid}$.
  \item On message $(\textsc{resume},\mathsf{eid},\mathsf{input})$ from a
        (registered) party $\party$, it calls the enclave
        $(\mathsf{idx},\mathsf{prog},\mathsf{mem})$ for $(\mathsf{eid},\party)$,
        where $\mathsf{mem}$ is the current memory state. It runs
        $\mathsf{prog}(\mathsf{input},\mathsf{mem})$ which returns
        $\mathsf{output}$ and an update memory state $\mathsf{mem}'$. Finally,
        it produces a signature $\sigma$ on
        $(\mathsf{idx},\mathsf{eid},\mathsf{prog},\mathsf{output})$ using
        $\mathsf{msk}$ and sends $(\mathsf{output},\sigma)$ to $\party$.
\end{itemize}

\subsection{The certification functionality $\funccert$}\label{app::func_Fcert}
We assume the existence of an ideal certification functionality $\funccert$, inspired by the certification functionality and the certification authority functionality introduced in~\cite{DBLP:conf/csfw/Canetti04}. The difference between $\funccert$ and the certification functionality in~\cite{DBLP:conf/csfw/Canetti04} is that (i) instead of taking over signature verification, $\funccert$ allows the verifier to verify the validity of a signature offline, and (ii) it allows the generation of only one certificate (signature) for each party.
 
In particular, the functionality $\funccert$ exposes methods $\certg$ and $\certs$. On the first call to $\certg$ via a message $\textsc{GetK}$ from a party $\party$, it initialises
an empty record and generates a signing keypair for signature scheme
$\sigscheme$, returning the verification key $\certvk$ on all subsequent calls to $\certg$. In a call to $\certs$, the input is a message $(\textsc{Sign},\mathsf{vk})$ from a party $\party$. The functionality checks that no other message has been
recorded from the same UC party id as $\party$. If this holds, then it returns the certificate $\certout$, which
is a signature on $\mathsf{vk}$ under the generated signing key. 

\subsection{The common reference string functionality $\CRS$}\label{app:func_crs}
The $\CRS$ functionality, as described in~\cite{PKC:BKMT21}, is parameterised by a distribution~$D$. On the first request message $\textsc{Get}$ from a party $P$, the functionality samples a CRS string $\mathsf{crs}$ from $D$ and sends $\mathsf{crs}$ to $P$. On any subsequent message $\textsc{Get}$, it returns the same string $\mathsf{crs}$.

\subsection{The secure channel functionality $\SC_R^S$}\label{app::func_sc}
We adopt the functionality $\SC_R^S$ from~\cite{PKC:BKMT21} that models a secure channel between sender $S$ and receiver $R$. The functionality keeps a record $M$ of the length of messages that are being sent. On message $(\textsc{send},m)$ from $S$, it sends $(\textsc{sent},m)$ to $R$ and appends the length of $m$ to the record $M$. The adversary is not activated upon
sending, but can later on request the record $M$. 
For simplicity, when the identity of the receiver or the sender is obvious, we will use the notation $\SC^S$ or $\SC_R$, respectively.

\subsection{The repository functionality $\REP$}\label{app::func_rep}
We relax the functionality $\REP$ from~\cite{PKC:BKMT21} by allowing any party to read/write, as long as the read/write request refers to some specified session. Namely, the functionality keeps a table $M$ of the all the messages submitted by writing requests. On message $(\textsc{write},x)$ from $W$, it runs the subroutine $\getHandle$ to obtain an identifying handle $\handle$, and records $x$ in $M[\handle]$. On message $(\textsc{read},\handle)$ from a party $\party\in\mathbf{R}$, it returns $M[\handle]$ to $P$.

\subsection{The time functionality $\timeidealfunc$}\label{app::func_time}
The time functionality $\timeidealfunc$ of~\cite{EPRINT:CKLRSSTVW20} can be used
as a clock within a UC protocol. It initialises a counter $t$ as $0$. On message
$\textsc{increment}$ from the environment, it increments $t$ by $1$. On message
$\textsc{time}$ from a party $\party$ , it sends $t$ to $\party$.

\subsection{The physical reality functionality $\phyidealfunc$}\label{app::func_phys}
Functionality $\phyidealfunc$ introduced in~\cite{EPRINT:CKLRSSTVW20},
represents the ``physical reality'' of each participant to a protocol, meaning
the historical record of all physical facts (e.g., location, motion, visible
surroundings) involving the participants.

\par\noindent $\phyidealfunc$ is parameterised by a validation predicate $V$ for
checking that the records provided by the environment are sensible, and a set
$\mathbf{F}$ of ideal functionalities that have full access to the records
obtained by $\phyidealfunc$. The functionality only considers records that have
a specific format and include the party identity, time, and the types of
measurement (e.g., location, altitude, temperature, distance of the party from
each other party, health status) that evaluate the physical reality for the said
party. It initialises a list $R$ of all submitted records that are in correct
format and operates as follows:
\begin{itemize}
\item On message $(\party,v)$ from the environment, where $\party$ is a party's
identity and $v$ is a record in correct format, it appends $(\party,v)$ to $R$.
Then, it sends $\textsc{time}$ to $\timeidealfunc$ (the time functionality presented
in~\ref{app::func_time}) and obtains $t$. It checks that $t$ matches the time entry in $v$ and that $V(R)$ holds. If any check fails, then it halts.
\item On message $(\textsc{MyCurrentMeas},\party,L,e)$ that comes from either party $\party$ or a functionality in $\mathbf{F}$ (otherwise, it returns an error), where $L$ is a list of fields that refer to the correct record format and $e$ is an error function:
\begin{enumerate}
\item It finds the latest entry $v$ in the sublist of entries in $R$ whose first element is $P$.
\item It sets $v_L$ as the record $v$ restricted to the fields in $L$.
\item It computes $e(v_L)$, i.e., the result of applying the error function $e$ to $v_L$.
\item It returns $e(v_L)$.
\end{enumerate}
\item On message $(\textsc{AllMeas},e)$ from a functionality in $\mathbf{F}$, it applies $e$ to each record in $R$ and obtains $\tilde{R}$. It returns $\tilde{R}$.
\end{itemize}

\subsection{The exposure notification functionality $\enidealfunc$}\label{app::func_en}
The Exposure Notification functionality, also introduced
in~\cite{EPRINT:CKLRSSTVW20}, builds on the previous two functionalities to
provide a mechanism for warning people who have been exposed to infectious
carriers of the virus. The description of the functionality is recapped in
section \ref{sec::enplus}; we show the formal description for the purposes of
comparing this functionality with $\enplusideal$.


Confirmation of test results when sharing exposure and re-registration into the
system for no longer infectious users is not captured by the functionality.

\vspace{1mm}

\begin{functionality}{$\enidealfunc[\enriskfunc,\enerrorfuncset,\enfakingfuncset,\enleakfunc]$}
  \begin{statedecl}
    \statevar{\ensharexprec}{}{List of users who have shared their exposure status}
    \statevar{\enusers}{}{List of active users}
    \statevar{\enuserscorr}{}{List of corrupted users}
    \statevar{\ennoisyrec}{}{Noisy record of physical reality}
  \end{statedecl}
  \begin{receive}[\adversary]{Setup}{\epsilon^{*}}
    \State \Assert $\epsilon^{*} \in \enerrorfuncset$; \hspace{1mm}
     $\ennoisyrec \gets \emptyset$ 
  \end{receive}
  \begin{receive}{ActivateMobileUser}{\enuser}
    \State $\enusers \gets \enusers \concat \enuser$
    \State \Send{ActivateMobileUser}{\enuser}{\adversary}
  \end{receive}
  \begin{receive}{ShareExposure}{\enuser}
    \State \SendAndReceive{AllMeas}{\epsilon^{*}}{\phyidealfunc}{\ennewnoisyrec}
    \State $\ennoisyrec \gets \ennoisyrec \concat \ennewnoisyrec$ 
    \If{$\ennoisyrec[\enuser][\phyrecinf] = \bot$}
    \State \Return$\enerrmsg$
    \Else

      \State \SendAndReceive{time}{}{\timeidealfunc}{\mathit{t}}
      \State $\ensharexprec \gets \ensharexprec \concat (\enuser,\mathit{t})$
      \State $\enusers \gets \enusers \setminus \{\enuser\}$
      \State \Send{ShareExposure}{\enuser}{\adversary}
    \EndIf
  \end{receive}
  \begin{receive}{ExposureCheck}{\enuser}
    \If{$\enuser \in \enusers$}
    \State \SendAndReceive{AllMeas}{\epsilon^{*}}{\phyidealfunc}{\ennewnoisyrec}
    \State $\ennoisyrec \gets \ennoisyrec \concat \ennewnoisyrec$ 
    \State $\mu \gets \ennoisyrec[\enuser] \concat \ennoisyrec[\ensharexprec]$
      \State \Return $\enriskfunc(\enuser,\mu)$
    \Else{}
      \Return$\enerrmsg$
    \EndIf
  \end{receive}
  \begin{receive}{RemoveMobileUser}{\enuser}
      \State $\enusers \gets \enusers \setminus \{\enuser\}$
  \end{receive}
  \begin{receive}[\adversary]{Corrupt}{\enuser}
    \State $\enuserscorr \gets \enuserscorr \concat \enuser$
  \end{receive}
  \begin{receive}[\adversary]{MyCurrentMeas}{\enuser,\phyrecfields,\phyerrfunc}
    \If{$\enuser \in \enuserscorr$}
    \State \SendAndReceive{MyCurrentMeas}{\enuser,\phyrecfields,\phyerrfunc}{\phyidealfunc}{\phymeasret}
    \State \Send{MyCurrentMeas}{\phymeasret}{\adversary}
    \EndIf
  \end{receive}
  \begin{receive}[\adversary]{FakeReality}{\phi}
    \If{$\phi \in \enfakingfuncset$}
      \State $\ennoisyrec \gets \phi(\ennoisyrec)$
    \EndIf
  \end{receive}
  \begin{receive}[\adversary]{Leak}{}
    \State \Send{Leak}{\enleakfunc(\{\ennoisyrec, \enusers, \ensharexprec\})}{\adversary}
  \end{receive}
  \begin{receive}[\environment]{IsCorrupt}{\enuser}
    \State \Return $\enuser \stackrel{?}{\in} \enuserscorr$
  \end{receive}
\end{functionality}

\subsection{The trusted bulletin board functionality $\tbbidealfunc$}\label{app::func_tbb}
The functionality $\tbbidealfunc$, as presented in~\cite{EPRINT:CKLRSSTVW20},
maintains a state that is updated whenever new data are uploaded (for infectious
parties). It initializes a list $C$ of records. On message $(\textsc{Add},c)$
from a party $\party$, it checks with $\phyidealfunc$ whether $\party$ is
infectious (formally, $\tbbidealfunc$ sends a message
$(\textsc{MyCurrentMes},\party,\mbox{ ``\textsf{health\_status}''},\mathsf{id})$
to $\phyidealfunc$, where $\mathsf{id}$ is the identity function). If this
holds, then it appends $c$ to $C$. On message \textsc{Retrieve} from a party
$\party$, it returns $C$ to $P$.


\section{Heatmap pseudocode}\label{app::heat}

In this section, we provide the pseudocode for the heatmap function discussed in Section \ref{sec::heatmap}. 
For simplicity of exposition, we assume that the input $\fein$ is already
a fully formed list of $\incubt \times \heatmapsize$ matrices containing a
single user's location data over the last $\incubt$ days. While $\glassvault$
functionalities typically expect a subset of $\phyidealfunc$'s noisy record of
reality for fields in $\phyrecsec$, turning those records in a list of matrices
can be delegated to the aggregator run by $\glassvault$ to turn individual
user's ciphertext into the multi-input list $\fein$.

\begin{algorithmic}
  \Function{$\heatmapfunc_{\heatmapparam}$}{$\fein,\heatmapstate$}
    \If{$\heatmapstate = \initstate$} $\heatmapmap \gets []$ \EndIf
    \For{$c \in \heatmapmap$}
        \State $c \gets c \concat \vec{0}$
        \If{$\forall i \in c: i = \vec{0}$} $\heatmapmap \gets \heatmapmap \setminus c$\EndIf
    \EndFor
    \For{$\heatmapinp\in\fein$}
        \State $b \gets \mathsf{CircularBuffer}(\incubt)$
        \For{$\{i=0;i<\incubt;i\plpl\}$}
          \State \Assert
          $\sum\limits_{\scriptscriptstyle j=0}^{\scriptscriptstyle \heatmapsize-1} \heatmapinp[i,j] = 24$
          \State $b \gets b \concat \heatmapinp[i,:]$
        \EndFor
        \State $\heatmapmap \gets \heatmapmap \concat b$
    \EndFor
    \State $\feout \gets \vec{0}$
    \If{$|\heatmapmap|\geq\heatmapmin$}
        \For{$\heatmapinp\in\heatmapmap$}\For{$\{i=0;i<\incubt;i\plpl\}$}\State $\feout \gets \feout + \heatmapinp[i,:]$\EndFor\EndFor
    \EndIf
    \State \Return $\feout$
  \EndFunction
\end{algorithmic}

The above pseudocode uses the following notation conventions:

\begin{itemize}
  \item Given matrix $z$, the notation $z[i,j]$ denotes accessing the
  $i$-th row and $j$-th column of $z$.
  \item $z[i,:]$ denotes the row vector corresponding to the $i$-th row of $z$;
  $z[:,j]$ is the column vector corresponding to the $j$-th column
  \item We denote by $\mathsf{CircularBuffer}(n)$ the creation of a new
  $n$-sized circular buffer. Appending an item to the buffer is accomplished
  through concatenation operator $\concat$. After $n$ items have been appended
  to a buffer, it will overwrite the first record in the buffer, and so on

\end{itemize}


\section{Steel simulator}\label{app::steel-sim}
\label{app:steelsim}
\vspace{-2mm}
We now describe the simulator presented in \cite{PKC:BKMT21}, while adapting the
message syntax to fit with the messages sent by $\dfesr$ and $\dsteel$.

\clearpage

\newcommand{\func}{\dfesrsim}

\begin{simulator}{$\IRONsim$[$\pke,\sigscheme,\nizkscheme,\secp,\Ffamily$]}
\vspace{-6mm}
  \begin{statedecl}
    \statevar{\simrepo[\cdot]}{\emptyset}{Table of ciphertext and handles in public repository}
    \statevar{\Flist}{[]}{List of $\prog{\FE^{\certvk}[\funcF]}$ enclaves and their $\eid{F}$}  
    \statevar{\gattsent}{\{\}}{Collects all messages sent to $\Gatt$ and its response}
    \statevar{\gattsigned}{\{\}}{Collects all messages signed by $\Gatt$}
    \statevar{(\ncrs,\trap)}{\nizkscheme.\Sim_1}{Simulated reference string and trapdoor}
  \end{statedecl}

 \noindent
\begin{receive}[\func]{setup}{\party}
  \If{$\mpk = \bot$}
  \State  $\eid{\KME} \gets \Gatt.\install(\PartyC.\sid, \prog{\KME^{\certvk}})$
  \State  $(\mpk, \irsig{\KME}) \gets \Gatt.\resume(\eid{\KME}, \initkw, \ncrs, \PartyC.\sid)$
  \EndIf

  \If{$\party = \PartyA$}
    \State  \Send{setup}{\mpk,$\hl{$\irsig{\KME}$}$}{\SC_{\PartyA}}
  \ElsIf{$P = \PartyB$}
	  \State \SendAndReceive{setup}{\mpk,$\hl{$\irsig{\KME}$}$, \eid{\KME}}{\SC^{\PartyC}_{\PartyB}}{\msg{provision}{\irsig{}, \eid{\DE}, \pkd}}
    \State \Assert{$(\PartyC.\sid, \eid{\DE}, \prog{\DE^{\certvk}}, \pkd) \in \gattsigned[\irsig{}]$}
	  \State  $(\ctsk, \irsig{\KME}) \gets \Gatt.\resume(\eid{\KME}, (\provisionkw,
    (\irsig, \eid{\DE}, \pkd, \eid{\KME},\ncrs))))$
	  \State  \Send{provision}{\ctsk, \irsig{\KME}}{\SC_{\PartyB}}
  \EndIf
  \end{receive}
%
%
\begin{receive}[ \text{party } \PartyB \text{ to } \REP]{read}{\handle}
  \State \SendAndReceive{decrypt}{\leakfunc,\handle}{\func \text{ on behalf of }
    \PartyB}{(\textsc{decrypted},|(\pkemsg,$\hl{$\lenkeyshare$}$)|)}
    \State \Assert{$|(\pkemsg,$\hl{$\lenkeyshare$}$)| \neq \bot$}
    \State $\ct \gets \pke.\pkee(\mpk, 0^{|(\pkemsg,\lenkeyshare)|})$
    \State $\pi \gets \nizkscheme.\Sim_2(\ncrs,\trap,(\mpk,\ct))$
    \State $\ctmsg \gets (\ct,\pi)$; $\simrepo[\ctmsg]\gets h$
    \State \Send{read}{\ctmsg}{\PartyB}
  \end{receive}
  \begin{receive}[ \text{party } \party\in\{\PartyBS\cup\PartyC\} \text{ to } \Gatt]{install}{\idx,\prog{}}
    \State  $\eid{} \gets \Gatt.\install(\idx, \prog{})$
    \State $\gattsent[\eid{}].\install \gets (\idx,\prog{},$\hl{$\party$}$)$
    \Comment{$\gattsent[\eid{}].install[1]$ is the program's code}
    \State \Forward $\eid{}$ to $\PartyB$
  \end{receive}
  \begin{receive}[ \text{party } \party\in\{\PartyBS\cup\PartyC\} \text{ to } \Gatt]{resume}{\eid{}, \inp}
    \State \Assert{$\gattsent[\eid{}]$\hl{$.\install[2]=\party$}}

  \If{$\gattsent[\eid{}].\install[1] \neq \prog{\FE^{\certvk}[\cdot]} \lor (\inp[0] \neq \runkw \lor \inp[-1]\neq\bot)$}

      \State  $(\outp, \irsig{}) \gets \Gatt.\resume(\eid{}, \inp)$
      \State $\gattsent[\eid{}].\resume \gets \gattsent[\eid{}].\resume \concat (\irsig{},\inp,\outp)$
      \State $\gattsigned[\irsig{}] \gets (\gattsent[\eid{}].\install[0], \eid{}, \gattsent[\eid{}].\install[1], \outp)$
      \State \Forward $(\outp, \irsig{})$ to $\party$
    \Else
      \State $(\idx, \prog{\FE^{\certvk}[\funcF]},$\hl{$\party$}$) \gets \gattsent[\eid{}].\install$
      \State $(\runkw, \irsig{\DE},\eid{\DE}, \ctsk, \ctmsg,$\hl{$\lenkeysharef$}$,\ncrs,\bot) \gets \inp$
      \State \Assert{$(\irsig{\funcF},(\initkw,\mpk,\idx), (\pkf)) \in \gattsent[\eid{}].\resume$}
      \State \Assert{$(\idx, \eid{}, \prog{\FE^{\certvk}[\funcF]}, \pkf) \in \gattsigned[\irsig{\funcF}] $}
      \State \Assert{$(\idx, \eid{\DE}, \prog{\DE^{\certvk}}, \ctsk,$\hl{$\lenkeysharef$}$,\ncrs)) \in \gattsigned[\irsig{\DE}]$}
      \Comment{If the ciphertext was not computed honestly and saved to $\simrepo$}
      \If{$\simrepo[\ctmsg] = \bot$}
      \State $(\ct,\pi) \gets \ctmsg$
      \State $((\pkemsg,$\hl{$\lenkeyshare$}$),\pkerc) \gets \nizkscheme.\Extr(\trap,(\mpk,\ct),\pi)$
      \State \SendAndReceive{encrypt}{\pkemsg,$\hl{$\lenkeyshare$}$}{\func \text{ on behalf of } \party}{\msg{encrypted}{\handle}}
      \If{$\handle\neq\bot$} $\simrepo[\ctmsg] \gets \handle$ \Else{ }\Return \EndIf
      \EndIf
      \State $\handle \gets \simrepo[\ctmsg]$
      \State \SendAndReceive{decrypt}{\funcF,\handle}{\func \text{ on behalf of } \party}{\msg{decrypted}{\feout}}
      \State $((\compkw, \feout), \irsig{} ) \gets \Gatt.\resume(\eid{\funcF}, (\runkw, \bot, \bot, \bot, \bot,$\hl{$\bot$}$,\bot,\feout))$
      \State $\gattsent[\eid{}].\resume \gets \gattsent[\eid{}].\resume \concat (\irsig{},\inp,(\compkw, \feout)))$
      \State $\gattsigned[\irsig{}] \gets (\gattsent[\eid{}].\install[0], \eid{}, \gattsent[\eid{}].\install[1], (\compkw, \feout))$
      \State \Forward $((\compkw, \feout),\irsig{})$ to $\party$
    \EndIf
  \end{receive}
\end{simulator}

\end{document}